\documentclass[aps,pra,twocolumn,groupedaddress,footinbib]{revtex4-1}
 
\usepackage{amsthm}
\usepackage{amssymb}
\usepackage{amsfonts}
\usepackage{graphicx}
\usepackage{amsmath}
\usepackage[sort&compress]{natbib}
\usepackage{mathrsfs}
\usepackage{enumerate}

\usepackage{eso-pic}
\usepackage{graphicx}
\usepackage{color}
\usepackage{type1cm}

\makeatother

\theoremstyle{definition} \newtheorem{kap}{Definition}
\newtheorem{thrm}{Theorem}
\newtheorem{lemmy}{Lemma}
\newtheorem{corol}{Corollary}

\newcommand{\mtxt}[1]
{
\ \ \textrm{#1} \ \ 
}

\newcommand{\bra}[1]
{
\langle #1|
}

\newcommand{\ket}[1]
{
| #1 \rangle
}

\newcommand{\pdot}{\stackrel{\LARGE \textbf{.}}{+}}
\newcommand{\tdot}{\stackrel{\LARGE \textbf{.}}{\times}}

\definecolor{dgreen}{rgb}{0.,0.6,0.}
\definecolor{JungleGreen}{cmyk}{0.99,0,0.52,0}
\definecolor{BlueGreen}{cmyk}{0.85,0,0.33,0}
 
\begin{document}

\title{A Generalization of the Functional Calculus of Observables and Notion of Joint Measurability to the Case of Non-commuting Observables}


\author{Richard DeJonghe}
 \author{Kimberly Frey}
 \author{Tom Imbo}

\affiliation{University of Illinois at Chicago, Department of Physics, 845 W. Taylor St., Chicago, IL 60607}

\date{\today}
\begin{abstract}

For \emph{any} pair of bounded observables $A$ and $B$ with pure point spectra, we construct an associated ``joint observable'' which gives rise to a notion of a joint (projective) measurement of $A$ and $B$, and which conforms to the intuition that one can measure non-commuting observables simultaneously, provided one is willing to give up arbitrary precision.  As an application, we show how our notion of a joint observable naturally allows for a construction of a ``functional calculus,'' so that for any pair of observables $A$ and $B$ as above, and any (Borel measurable) function $f: \mathbb{R}^2 \rightarrow \mathbb{R}$, a new ``generalized observable'' $f(A,B)$ is obtained.  Moreover, we show that this new functional calculus has some rather remarkable properties. 

\end{abstract}



\maketitle

\section{Introduction}

In quantum theory observable quantities are represented by Hermitian operators, and the orthodox view is that two observables are simultaneously measurable if and only if their corresponding Hermitian operators commute.  However, to focus on a commonly used example, one also frequently hears that the uncertainty relation $\Delta x \Delta p \gtrsim \hbar$ should be taken to mean that one \emph{can} simultaneously measure position and momentum, \emph{just not with arbitrary precision}.  Beginning with Arthurs and Kelly~\cite{Arthurs_64}, various researchers have investigated the possibility of measuring non-commuting observables, typically in the framework of POVMs (i.e.~\emph{unsharp} measurements) \cite{She_66, Park_68, Yuen_82, Busch_85, Raymer_94, Busch_96, Busch_07, Busch_13},  substantiating the above connection between the uncertainty principle and joint measurability.

In this paper we will discuss a new notion of a \emph{sharp} joint measurement of a pair of bounded observables with pure point spectra~\footnote{However, some of our results are applicable to a wider class of observables.}.  This joint measurement also conforms to the intuition that one can measure non-commuting observables simultaneously, provided one is willing to give up arbitrary precision.  The way in which we arrive at this notion of joint measurability is via the construction of joint observables --- indeed, our notion of a joint observable generalizes that of Gudder \cite{Gudder_68} and Varadarajan \cite{Varadarajan_62} to the case of non-commuting observables.

Additionally, we find that our notion of a joint observable naturally allows for the construction of a ``functional calculus of observables'' --- i.e. for any (Borel measurable) function $f: \mathbb{R}^2 \rightarrow \mathbb{R}$, and for any observables $A$ and $B$ as above, we have a natural way of defining an object $f(A,B)$, which can be thought of as a ``generalized observable''~\footnote{If the spectrum of either $A$ or $B$ is not finite (and $[A,B] \neq 0$), we require that $f:\mathbb{R}^2 \to \mathbb{R}$ be continuous. }.  Of course, when $A$ and $B$ commute, one already has recourse to a fully developed functional calculus of observables.  In particular, one can define $f(A,B)$ either using the spectral decompositions of $A$ and $B$ (see, e.g. \cite{Kreyszig}) or by using Gudder's aforementioned notion of a joint observable \cite{Gudder_68, Varadarajan_62}.  As may be expected, when the function $f$ is simply addition or multiplication of numbers, applying $f$ to commuting observables in the manner just described yields the usual linear algebraic sum or product, respectively.

Our generalized functional calculus agrees with the ordinary functional calculus when the observables commute.  For non-commuting observables $A$ and $B$, there is associated with each generalized observable $f(A,B)$, a natural family of ordinary observables (parameterized by $\mathcal{E}$ as in section \ref{sec:bp-functionalcalculus}).  Each member $f_{\mathcal{E}}(A,B)$ of this family has the following properties (among others):

\begin{enumerate}[(1)]
\item A simultaneous eigenstate of $A$ and $B$ with eigenvalues $a,b$ respectively, is an eigenstate of $f_{\mathcal{E}}(A,B)$ with eigenvalue $f(a,b)$.
\item Every element of the spectrum of $f_{\mathcal{E}}(A,B)$ is of the form $f(a,b)$, where $a,b$ are in the spectra of $A$ and $B$, respectively.
\end{enumerate}
While property (1) above is perhaps expected for any reasonable construction (and indeed is satisfied for the ordinary linear algebraic sum $A+B$), property (2) is a novel feature of our functional calculus --- compare this to the complicated relationship between the spectrum of, e.g.~$A+B$ and the spectra of $A$ and $B$ when ${[A,B] \neq 0}$.  As such, it is immediately clear that the operations corresponding to addition in this functional calculus will not reproduce the usual linear algebraic sum whenever $A$ and $B$ do not commute.

This paper is organized as follows.~In sections~${\textrm{\ref{sec:commuting}-\ref{sec:conclusion}}}$ we restrict our attention to finite-dimensional Hilbert spaces~\footnote{All observables on a finite-dimensional Hilbert space have a finite number of eigenvalues, and hence are bounded with pure point spectra.}.  In section \ref{sec:commuting} we review projection lattices, projection-valued measures, and Gudder's notion of a joint observable in the commuting case.  In section \ref{sec:generalizedJOs} we present our notion of a joint observable as well as demonstrate a schema which realizes the corresponding joint measurement in the framework of quantum operations, and connect our joint observable to the uncertainty principle.  As an application, in section \ref{sec:funcs} we use this notion of a joint observable to construct our functional calculus and describe its properties.  We follow with a simple example of the new functional calculus ``in action''.  Finally, in section \ref{sec:conclusion} we conclude with some brief remarks.  Proofs and further technical results in the more general case of observables on (possibly) infinite-dimensional Hilbert spaces can be found in the appendix.  In what follows (except in the appendix), $\mathcal{H}$ will denote a finite-dimensional complex Hilbert space.  We  will further take all functions $f:\mathbb{R}^2 \to \mathbb{R}$ to be Borel measurable, and all partitions of measurable spaces to consist of measurable sets.

\section{Background}
\label{sec:commuting}

\subsection{Observables, Spectral Families, and Projection-Valued Measures}
\label{sec:kapPVMs}

There are many different ways to represent the observables of a quantum system.  The usual way is as a self-adjoint linear operator $A$ on $\mathcal{H}$.  By the spectral theorem $A = \sum_{i = 1}^n \lambda_i P_i$, where $P_i$ is the projector onto the eigenspace of $\mathcal{H}$ associated with the eigenvalue $\lambda_i$ of $A$, and $P_i \perp P_j$ (i.e.~$P_i P_j = P_j P_i = 0$) when $i \neq j$.  Without loss of generality, we assume that $\lambda_i < \lambda_{i+1}$ for all $i \in \{ 1, \ldots,  n \}$.   Another way of representing observables is in terms of a spectral family --- i.e. a one parameter family $\{ E_{\lambda} \}_{\lambda \in \mathbb{R}}$ of projection operators on $\mathcal{H}$ such that (i)~$E_{\lambda} \leq E_{\mu}$ whenever $\lambda \leq \mu$, (ii)~lim$_{\lambda \rightarrow -\infty} E_{\lambda} = 0$, and (iii)~lim$_{\lambda \rightarrow +\infty} E_{\lambda} = I$, where $0$ is the zero operator and $I$ is the identity operator on $\mathcal{H}$, and $\leq$ is the partial order on projectors defined by $E_{\lambda} \leq E_{\mu}$ whenever the respective subspaces $\Sigma_{\lambda}, \Sigma_{\mu}$ onto which they project satisfy $\Sigma_{\lambda} \subseteq \Sigma_{\mu}$. 
(Note that $0 \leq P \leq I$ for every projection operator $P$.)  
There is a 1-1 correspondence between such spectral families of projectors and self-adjoint operators on $\mathcal{H}$ \footnote{\label{fn:specfam} Actually, one additional requirement is needed to obtain this 1-1 correspondence, namely right continuity ($\lim_{\lambda' \to \lambda^+}E_{\lambda^{\prime}} = E_{\lambda}$ for all $\lambda \in \mathbb{R}$).  A spectral family not satisfying this requirement still corresponds to a unique self-adjoint operator, but each self-adjoint operator corresponds to multiple spectral families when right-continuity is not imposed.  Of course, given any spectral family $\{ E_\lambda \}_{\lambda \in \mathbb{R}}$, one can form a right-continuous spectral family which corresponds to the same self-adjoint operator, namely $\{ \hat{E}_\lambda \}_{\lambda \in \mathbb{R}}$, where $ \displaystyle \hat{E}_\lambda := \lim_{\lambda' \to \lambda^+} E_{\lambda'}$.   Finally, although we restrict our discussion to the finite-dimensional case in the initial sections of this paper, properties (i)-(iii) above define a spectral family on \emph{any} separable complex Hilbert space $\mathcal{H}$, and a similar 1-1 correspondence between (now possibly unbounded) self-adjoint operators on $\mathcal{H}$ and (right continuous) spectral families on $\mathcal{H}$ also holds in this case.}.  Namely, for the self-adjoint operator $A$ given above, and for $\lambda \in \mathbb{R}$, the elements of the corresponding spectral family are given by $E_{\lambda} = \sum_{\lambda_i \leq \lambda} P_i$.  Conversely, given a spectral family $\{ E_{\lambda} \}_{\lambda \in \mathbb{R}}$, its associated self-adjoint operator is given by $\sum_{\lambda \in \mathbb{R}} \lambda \Big( E_{\lambda} - \textrm{lim}_{\epsilon \rightarrow 0^+} E_{\lambda - \epsilon} \Big)$.

Alternatively, one can describe observables on $\mathcal{H}$ in terms of \emph{projection-valued measures} (PVMs), which have an elegant formulation using the projection lattice $\mathcal{L}_{\mathcal{H}}$ of~$\mathcal{H}$ (where $\mathcal{L}_{\mathcal{H}}$ is the set of all projection operators on~$\mathcal{H}$ equipped with the partial order $\leq$ defined above), and which we describe below after some preliminaries about the projection lattice.  $\mathcal{L}_{\mathcal{H}}$ is a \emph{complete lattice}, which is to say that for any subset of projectors $\{ P_j \}_{j \in J} \subseteq \mathcal{L}_{\mathcal{H}}$, there exists both a least upper bound and greatest lower bound in $\mathcal{L}_{\mathcal{H}}$, which we denote by $\bigvee_{j \in J} P_j$ and $\bigwedge_{j \in J} P_j$, respectively~\cite{Holland_75}.  (For pairs of projectors $P_1$ and $P_2$, we use $P_1 \vee P_2$ and $P_1 \wedge P_2$ to denote their least upper and greatest lower bounds.)  Furthermore, to each projector~$P$ there corresponds a unique projector $P^\perp$ which satisfies ${P \perp P^\perp}$ (and hence ${P \wedge P^\perp = 0}$) and ${P \vee P^\perp = I}$, namely ${P^\perp = I - P}$.  Finally,~for~$\{P_i \}_{i=1}^\infty$ a set of pairwise orthogonal projectors, we have $\bigvee_{i =1}^\infty P_i = \sum_{i=1}^{\infty} P_i$, and for commuting projectors $P_1$ and $P_2$, we have $P_1 \wedge P_2 = P_1P_2$.

We are now in a position to define a PVM on $(\Omega, \mathcal{M})$, where $(\Omega, \mathcal{M})$ is a \emph{measurable space} (that is, $\Omega$ is a set and $\mathcal{M}$ is a Boolean $\sigma$-algebra of subsets of $\Omega$ whose elements are called \emph{measurable sets}).  A PVM on $(\Omega, \mathcal{M})$ is a \emph{$\sigma$-homomorphism} from $\mathcal{M}$ to $\mathcal{L}_{\mathcal{H}}$ --- i.e.~a map ${\alpha: \mathcal{M} \rightarrow \mathcal{L}_{\mathcal{H}}}$ which satisfies 
\begin{enumerate}[(i)] 
\item $\alpha(\Omega) = I$;
\item if $R_1,R_2 \in \mathcal{M}$ are such that $R_1 \cap R_2 = \emptyset$, then $\alpha(R_1) \bot \, \alpha(R_2)$;
\item $\displaystyle \sum_{i=1}^\infty \alpha(R_i) = \alpha \Big( \bigcup_{i=1}^\infty R_i \Big)  $ for all $R_1, R_2, \ldots \in \mathcal{M}$  such that $R_i \cap R_j = \emptyset$ whenever $i \neq j$. 
\end{enumerate}
The set of observables on $\mathcal{H}$ are in 1-1 correspondence with the PVMs on $(\mathbb{R},\mathcal{B}(\mathbb{R}))$,  where $\mathcal{B}(\mathbb{R})$ denotes the Borel subsets of~$\mathbb{R}$.  Explicitly, for a given self-adjoint operator $A$, the associated PVM $\alpha_A$ is given by $\alpha_A(R) = \sum_{\lambda_i \in R} P_i$ (for any $R \in \mathcal{B}(\mathbb{R})$), where the $P_i$'s are the projectors in the spectral decomposition of $A$; conversely, given a PVM $\alpha$, there is a unique observable determined by the spectral family whose elements are defined by $E_{\lambda} := \alpha \big((-\infty, \lambda] \big)$.  In the sequel, we will use the same symbol (as well as the term `observable') to refer to any of these three ways of representing an observable, as this standard abuse of notation enables us to streamline the following discussion.  For example, for an observable $A$, and any $S \in \mathcal{B}(\mathbb{R})$, we will simply write $A(S)$ instead of $\alpha_A(S)$.

\subsection{Coarse-graining of Observables}

Given an observable $A$, any set $S \in \mathcal{B}(\mathbb{R})$ naturally corresponds to a two-outcome measurement associated with the projection operators $A(S)$ and $A(S)^{\bot} = A(S^c)$, where $S^c$ is the set theoretic complement of $S$.  For a system in the state represented by the density operator $\rho$, an unselected measurement of $A$ associated with these two outcomes takes $\rho$ to ${ A(S) \, \rho \, A(S) \, + \, A(S^c) \, \rho \, A(S^c) }$.  Since the projection operator $A(S)$ corresponds to the subspace spanned by all states $\ket{\psi} \in \mathcal{H}$ such that one can say with certainty that the value of the observable $A$ is in the range $S$ (i.e.~a measurement of $A$ yields a value in $S$ with probability $1$), and similarly for $A(S^c)$, we see that one can think of $A(S)$ and $A(S^c)$ together as a ``coarse-graining'' of $A$ whereby one only determines which \emph{region} $A$ takes its value in upon measurement (either $S$ or $S^c$), but not the specific value.  In fact, since the set of ``outcomes'' of this coarse-grained measurement consists of $\mathcal{P} := \{S,S^c\}$, one can define a course-grained observable $\tilde{A}$ associated with $A$ more precisely as a PVM on the measurable space  $(\mathcal{P}, 2^{\mathcal{P}})$, where $2^{\mathcal{P}}$ is the set of all subsets of $\mathcal{P}$.  In particular, we define $\tilde{A}(Q) := A(\bigcup Q)$, where $Q \in 2^{\mathcal{P}}$ (and $\bigcup X := \bigcup_{Z \in X} Z $ for any set $X$ whose elements $Z$ are, themselves, sets).  This procedure can be generalized, allowing $A$ to be a PVM on any measurable space $(\Omega, \mathcal{M})$, and $\mathcal{P}$ to be any partition of $\Omega$ --- each such partition will be associated with a coarse-grained PVM $\tilde{A}$~\footnote{In this case, one must replace $2^{\mathcal{P}}$ with an appropriate collection $\mathcal{M}_\mathcal{P}$ of subsets of $\mathcal{P}$ --- see Theorem \ref{thm:coarse-grain} in the appendix for details in a more general context.}.  Since $\tilde{A}$ is determined uniquely by $A$ and $\mathcal{P}$, any observable contains complete information about all its possible coarse-grainings.  Additionally, note that since unselected measurements are a special case of trace preserving quantum operations \cite{Nielsen} and each partition of the outcome space is associated with an unselected measurement of the (course-grained) observable, we see that the description of measurements of these observables fits naturally into the more general framework of quantum operations.

\subsection{Joint Observables in the Commuting Case}

If $A$ and $B$ are observables on $\mathcal{H}$, a PVM ${K:\mathcal{B}(\mathbb{R}^2) \rightarrow \mathcal{L}_{\mathcal{H}}}$ (where $\mathcal{B}(\mathbb{R}^2)$ denotes the Borel subsets of $\mathbb{R}^2$) which satisfies $K(R_1 \times R_2) = A(R_1) \wedge B(R_2)$ for all $R_1, R_2 \in \mathcal{B}(\mathbb{R})$ is said to be a \emph{joint observable for $A$ and $B$}.  It follows immediately that any joint observable $K$ for $A$ and $B$ has the property that $K(R_1 \times \mathbb{R}) = A(R_1)$ and $K(\mathbb{R} \times R_2) = B(R_2)$, which is to say that $K$ has the expected \emph{margins}.
If we suppose that $A$ and $B$ commute, it is straightforward to show that the map $J_{AB}: \mathcal{B}(\mathbb{R}^2) \rightarrow \mathcal{L}_{\mathcal{H}}$ defined by (for $Q \in \mathcal{B}(\mathbb{R}^2)$) 
\begin{equation}
\label{J_AB}
J_{AB}(Q) := \bigvee_{\substack{R_1 \times R_2 \subseteq Q \\ R_1 \times R_2 \in \mathcal{B}(\mathbb{R}^2)}} \big[ A(R_1) \wedge B(R_2)\big]
\end{equation}
is a joint observable.  Gudder \cite{Gudder_68} was the first to show that a joint observable for $A$ and $B$ exists if and only if $[A, B ] = 0$, and moreover, that this joint observable is unique.  (However, as far as we can tell, expression (\ref{J_AB}) above for this unique joint observable has not previously appeared in the literature.)

For any $R_1 \times R_2 \in \mathcal{B}(\mathbb{R}^2)$ and for any observables $A$ and $B$, we have that $A(R_1) \wedge B(R_2)$ is the projection operator onto the subspace spanned by the set of all states $| \psi \rangle \in \mathcal{H}$ such that if the system is initially in the state $| \psi \rangle$, a measurement of the observable $A$ yields (with certainty) an outcome in $R_1$ and a measurement of the observable $B$ yields (with certainty) an outcome in $R_2$.  Given this, we see that when $A$ and $B$ commute, $J_{AB}$ encodes all information about a simultaneous measurement of $A$ and $B$, as well as information about all possible ``coarse-grainings'' of this simultaneous measurement associated with partitions of the outcome space $\mathbb{R}^2$.  This is to say that for any partition $\{ Q_1, \ldots, Q_n \}$ of $\mathbb{R}^2$ with $Q_i \in \mathcal{B}(\mathbb{R}^2)$ for all $i$, we have that $\sum_{i=1}^n J_{AB}(Q_i) = I$, and the (trace preserving) quantum operation 
\begin{equation}
\label{quantumoperationeq}
\rho \rightarrow \sum_{i=1}^n J_{AB}(Q_i) \, \rho \, J_{AB}(Q_i)
\end{equation}  
describes an (unselected and course-grained) simultaneous measurement of $A$ and $B$.  

Further, for any function $f:~\mathbb{R}^2~\rightarrow~\mathbb{R}$ and pair of commuting observables $A$ and $B$, one obtains the aforementioned functional calculus of (commuting) observables by taking the PVM $f(A,B)$ to be defined in terms of the joint observable by $f(A,B) := J_{AB} \circ f^{-1}$~\cite{Gudder_68}, where ${ f^{-1}: \mathcal{B}(\mathbb{R}) \rightarrow \mathcal{B}(\mathbb{R}^2) }$ is the map which takes each $X \in \mathcal{B}(\mathbb{R})$ to its \emph{pre-image} under $f$.   As mentioned previously, when the function $f$ is simply addition or multiplication of numbers, applying $f$ to commuting observables in the manner just described yields the usual linear algebraic sum or product, respectively.

\subsection{Example: $J_{AB}$ when $\dim \mathcal{H} = 2$ and $[A,B]=0$}
\label{sec:ex2DJAB-commuting}

In what follows, we take $\mathcal{H}$ to be a two-dimensional Hilbert space.  Also, let $A= \alpha I + \mathbf{a} \cdot \boldsymbol{\sigma}$ and $B = \beta I + \mathbf{b} \cdot \boldsymbol{\sigma}$, where $\alpha, \beta \in \mathbb{R}$, $\mathbf{a}, \mathbf{b} \in \mathbb{R}^3$, $I$ is the identity matrix, and $\boldsymbol{\sigma}$ is the 3-vector consisting of the Pauli matrices.  Note that the eigenvalues of $A$ and $B$ are $a_{\pm} := \alpha \pm | \mathbf{a}|$ and $b_{\pm} := \beta \pm | \mathbf{b} |$, respectively, and assume $[A,B]=0$ (which is equivalent to $\mathbf{a}$ and $\mathbf{b}$ being co-linear).

In the case in which both $A$ and $B$ each have a single eigenvalue (i.e.~$A = \alpha I $ and $B = \beta I$), we have that for any $Q \in \mathcal{B}(\mathbb{R}^2)$,
\begin{equation}
J_{AB}(Q) = \begin{cases} \ I \mbox{ if } (\alpha,\beta) \in Q \\
                         \ 0 \mbox{ if } (\alpha,\beta) \notin Q. \end{cases} 
\end{equation}

The next case to consider is when only one of the observables, say $A$, has two distinct eigenvalues.  In this case, the projector onto the eigenspace of $A$ with eigenvalue $a_+ $ is given by $P_+ := \frac{1}{2}(I + \hat{\mathbf{a}} \cdot \boldsymbol{\sigma})$, where  $\hat{\mathbf{a}} = \mathbf{a}/|\mathbf{a}|$.  Similarly the projector on to the eigenspace of $A$ with eigenvalue $a_-$ is given by $P_- := \frac{1}{2}(I - \hat{\mathbf{a}} \cdot \boldsymbol{\sigma}) = P^\perp_+$.  Then, for any $Q \in \mathcal{B}(\mathbb{R}^2)$, it is easy to see that
\begin{equation}
J_{AB}(Q) = \begin{cases} \ I \mbox{  if both } (a_+,\beta), (a_-,\beta) \in Q \\
                         \ P_+ \mbox{ if } (a_+,\beta) \in Q \mbox{ and }  (a_-, \beta) \notin Q \\
                         \ P_- \mbox{ if } (a_+,\beta) \notin Q \mbox{ and }  (a_-, \beta) \in Q \\
                         \ 0 \mbox{  if both } (a_+,\beta), (a_-,\beta) \notin Q. \end{cases} 
\end{equation}

Moving on to the case in which both $A$ and $B$ have two distinct eigenvalues, there are two possibilities corresponding to whether the eigenstate $\ket{a_+}$ of $A$ with eigenvalue $a_+$ is an eigenstate of $B$ with eigenvalue $b_+$ or $b_-$.  We proceed assuming $B \ket{a_+} = b_+ \ket{a_+}$; analysis of the other possibility proceeds analogously.  Notice that in this case $P_+$, as defined above, is also the projector onto the eigenspace of $B$ with eigenvalue $b_+$.  A straightforward computation then shows (for $Q \in \mathcal{B}(\mathbb{R}^2)$)

\begin{equation}
J_{AB}(Q) = \begin{cases} \ I \mbox{  if both } (a_+,b_+), (a_-,b_-) \in Q \\
                         \ P_+ \mbox{ if } (a_+,b_+) \in Q \mbox{ and }  (a_-, b_-) \notin Q \\
                         \ P_- \mbox{ if } (a_+,b_+) \notin Q \mbox{ and }  (a_-, b_-) \in Q \\
                         \ 0 \mbox{  if both } (a_+,b_+), (a_-,b_-) \notin Q. \end{cases} 
\end{equation}

Note that (as expected) $J_{AB}$ is a PVM in each of the three cases discussed above.  Also, it is easy to see that $J_{AB}(R_1 \times R_2) = A(R_1) \wedge B(R_2)$ for any $R_1,R_2 \in \mathcal{B}(\mathbb{R})$, so that $J_{AB}$ is, in fact, a joint observable.

\section{Joint Observables in the Non-commuting Case}
\label{sec:generalizedJOs}

\subsection{Generalized Joint Observables and Joint Measurability}
\label{sec:JOs&JM}

We will now demonstrate that for \emph{any} pair of observables $A$ and $B$, the expression for $J_{AB}$ in (\ref{J_AB}) above, which is still well-defined when $[A,B] \neq 0$, has a natural interpretation in terms of measurement even though it is no longer a PVM in this case. First note that $J_{AB}$ still has the correct margins, even when $A$ and $B$ don't commute. It is also straightforward to show that properties (i) and (ii) of PVMs (defined in section \ref{sec:kapPVMs}) still hold.  (See Theorem \ref{thm:joint_pvim} in the appendix for a proof.)  However, while property (iii), also known as \emph{countable additivity}, need not hold in general,  $J_{AB}$ does satisfy \emph{countable sub-additivity}, which is to say that for all $R_1, R_2, \ldots \in \mathcal{B}(\mathbb{R}^2)$ such that $R_i \cap R_j = \emptyset$ ($i \neq j$), we have 
\begin{equation}
\label{subadditivityofJ_AB}
\displaystyle \sum_{i=1}^\infty J_{AB}(R_i) \leq J_{AB} \Big( \bigcup_{i=1}^\infty R_i \Big) .  
\end{equation}


Notice that for any partition $\{ Q_1, \ldots, Q_n \}$ of $\mathbb{R}^2$, $\bigvee_{i=1}^{n} J_{AB}(Q_i) = \sum_{i=1}^n J_{AB}(Q_i)$ holds as a consequence of property (ii), but that this sum of orthogonal projectors need not equal the identity operator on $\mathcal{H}$ due to the failure of countable additivity (although $\sum_{i=1}^n J_{AB}(Q_i) \leq I$ by equation (\ref{subadditivityofJ_AB}) above).  
Hence, to any partition $\{ Q_1, \ldots, Q_n \}$ of $\mathbb{R}^2$ with ${\sum_{i=1}^n J_{AB}(Q_i) \neq I}$, there corresponds an unselected measurement in which there is a chance that we do not obtain any of our measurement outcomes --- that is, the corresponding quantum operation (as in equation (\ref{quantumoperationeq})) is not trace preserving.  \emph{It is this quantum operation which we will refer to as a simultaneous (or joint) measurement of $A$ and $B$} (independent of whether or not $\sum_{i=1}^n J_{AB}(Q_i) = I$ for the partition $\{ Q_1, \ldots, Q_n \}$ of $\mathbb{R}^2$, and independent of whether or not $[A,B] = 0$).  Note also that such measurements are \emph{sharp} (albeit course-grained) since the $J_{AB}(Q_i)$'s are pairwise orthogonal projection operators.  We will refer to $J_{AB}$ above as a \emph{generalized joint observable}.

We now give an explicit realization of such an unselected measurement (as a combination of unitary evolution and selected measurement) using an ancilla system.  Let $\{ Q_1, \ldots, Q_n \}$ be a partition of $\mathbb{R}^2$.  The unselected simultaneous measurement of $A$ and $B$ associated with this partition is given by the following schema.  First, define 
\begin{equation}
J_{AB}^0 := I - \sum_{i=1}^n J_{AB}(Q_i), 
\end{equation}
and let $\mathcal{A}$ be an $n+1$ dimensional Hilbert space with orthonormal basis $\{ \ket{i} \}_{i=0}^n$.  From equation (\ref{subadditivityofJ_AB}), it is easy to see that for any $\ket{\psi} \in \mathcal{H}$, there exists a unitary operator $U$ on $\mathcal{H} \otimes \mathcal{A}$ satisfying 
\begin{equation}
U \big( \ket{\psi} \otimes \ket{0} \big) = J_{AB}^0 \ket{\psi} \otimes \ket{0} + \sum_{i=1}^n \big( J_{AB}(Q_i) \ket{\psi} \big) \otimes \ket{i}. 
\end{equation}
Starting with an initial state $\rho$ on $\mathcal{H}$, form the state $\rho^{\prime} := \rho \otimes \ket{0} \bra{0}$ on $\mathcal{H} \otimes \mathcal{A}$.  Then evolve the state as $\rho^{\prime} \mapsto U \rho^{\prime} U^\dag$, and, following this, projectively measure the operator $I \otimes \sum_{i=1}^n \ket{i} \bra{i}$, selecting for the $+1$ eigenvalue.  Finally, trace over $\mathcal{A}$.  It is straightforward to see that this gives the evolution in equation (\ref{quantumoperationeq}), but now where $\sum_{i=1}^n J_{AB}(Q_i) \neq I$ in general.

\subsection{Connection to the Uncertainty Principle}

The extent to which the above quantum operation manages to be trace preserving increases, in general, with more coarse-graining of our partitions, due to the sub-additivity of $J_{AB}$.  We can interpret this as a manifestation of the uncertainty principle with regard to our joint measurements --- the essential feature is that as we decrease the resolution of the measurement, it becomes easier to find states for which we can say that the values of $A$ and $B$ for that state are constrained to lie in any fixed region of the plane.  In fact, for any state $\ket{\psi} \in \mathcal{H}$, we can make a direct quantitative connection between the uncertainty principle and any convex rectangular region $R_1 \times R_2 \in \mathcal{B}(\mathbb{R}^2)$ for which $J_{AB}(R_1 \times R_2) \ket{\psi} = \ket{\psi}$.

As usually stated, for a system in the state $\ket{\psi} \in \mathcal{H}$, the uncertainty principle puts a lower bound on the product of the (square roots of the) variances of the outcomes of any pair of observables.  For example, the Robertson relation \cite{Robertson_29} for the observables $A$ and $B$ is
\begin{equation}
\Delta A \Delta B \geq \frac{1}{2} | \langle [A,B] \rangle |,
\end{equation}
where for any self-adjoint operator $Z$ on $\mathcal{H}$, we have that ${ \langle Z \rangle := \bra{\psi} Z \ket{\psi} }$, as well as that ${ \Delta Z := \sqrt{\langle Z^2 \rangle - \langle Z \rangle^2} }$.  
Now, for a given state ${ \ket{\psi} \in \mathcal{H} }$ and any convex region ${ R_1 \times R_2 \subseteq \mathcal{B}(\mathbb{R}^2) }$ such that ${ J_{AB}(R_1 \times R_2) \ket{\psi} = \ket{\psi} }$, we have 
\begin{equation}
\label{ourUncertaintyREL}
\frac{1}{2} \textrm{Area}(R_1 \times R_2) \geq  \Delta A \Delta B
\end{equation}   
(see Theorem \ref{thm:uncertaintyREL} in the appendix).  The minimal such value of ${ \textrm{Area}(R_1 \times R_2) }$ can thus be thought of as a measure of how ``incompatible'' $A$ and $B$ are, or of how uncertain a joint measurement of $A$ and $B$ is, in the state $\ket{\psi}$.  Interestingly, various investigations of joint measurements of non-commuting observables in the unsharp (POVM) case \emph{also} find that their natural measures of the uncertainty of the joint measurement (the analog of our minimal Area($R_1 \times R_2$) above) are bounded below exactly as in inequality (\ref{ourUncertaintyREL}) \cite{Arthurs_64, She_66, Yuen_82, Raymer_94}.

\subsection{Generalized Projection-Valued Measures}

Although $J_{AB}$ is \emph{not} a PVM when $A$ and $B$ do not commute, it comes ``close'' in the sense that it satisfies properties (i) and (ii) and is countably sub-additive (as noted previously).  In the sequel, for any measurable space $(\Omega, \mathcal{M})$, a map $\alpha: \mathcal{M} \rightarrow \mathcal{L}_{\mathcal{H}}$ which satisfies properties (i) and (ii) of PVMs, along with countable sub-additivity 
\begin{align*}
\textrm{(iii}^{\prime} ) \quad & \sum_{i=1}^\infty \alpha(R_i) \leq \alpha \Big( \bigcup_{i=1}^\infty R_i \Big)  \mtxt{for all} R_1, R_2, \ldots \in \mathcal{M} \\ & \textrm{such that} \ \ R_i \cap R_j = \emptyset \mtxt{whenever} i \neq j
\end{align*}
will be called a \emph{generalized projection-valued measure} (gPVM) on $(\Omega, \mathcal{M})$.  Notice that any gPVM $\alpha$ on $(\Omega, \mathcal{M})$ also satisfies $\alpha (\emptyset) = 0$, and that if $Q,S \in \mathcal{M}$ are such that $Q \subseteq S$, then $\alpha(Q) \leq \alpha(S)$ --- that is, gPVMs are monotonic and increasing.  
We now proceed to investigate further properties of the gPVM $J_{AB}$.

\subsection{Example: $J_{AB}$ when $\dim \mathcal{H} = 2$ and $[A,B] \neq 0$}
\label{sec:ex2DJAB-noncommuting}

As in section \ref{sec:ex2DJAB-commuting} we take $\dim \mathcal{H} = 2$, and let ${A= \alpha I + \mathbf{a} \cdot \boldsymbol{\sigma}}$ and ${B = \beta I + \mathbf{b} \cdot \boldsymbol{\sigma}}$, with eigenvalues $a_{\pm}$ and $b_{\pm}$, respectively.  

We begin by noting that when $[A,B] \neq 0$, $A$ and $B$ each have two distinct eigenvalues.  We retain the definition of $P_{\pm} = \frac{1}{2}(I \pm \hat{\mathbf{a}} \cdot \boldsymbol{\sigma})$ (from section \ref{sec:ex2DJAB-commuting}), and also define $Q_{\pm} := \frac{1}{2}(I \pm \hat{\mathbf{b}} \cdot \boldsymbol{\sigma})$, so that $Q_{\pm}$ is the projector onto the eigenspace of $B$ with eigenvalue $b_{\pm}$.  The computation of $J_{AB}$ yields, for any $Q \in \mathcal{B}(\mathbb{R}^2)$ (where $\sigma(A)$ denotes the spectrum of $A$)
\begin{align}
\label{eq:JAB-noncommuting}
J& _{AB}  (Q) = \nonumber \\
& \begin{cases} \ I \mbox{  if at least 3 elements of } \sigma(A) \times \sigma(B) \mbox{ are in } Q \\
                         \ Q_+ \mbox{ if both } (a_{\pm},b_+) \in Q, \mbox{ and both }  (a_{\pm},b_-) \notin Q \\
                         \ Q_- \mbox{ if both } (a_{\pm},b_+) \notin Q, \mbox{ and both }  (a_{\pm},b_-) \in Q \\
                         \ P_+ \mbox{ if both } (a_+,b_{\pm}) \in Q, \mbox{ and both }  (a_-,b_{\pm}) \notin Q \\
                         \ P_- \mbox{ if both } (a_+,b_{\pm}) \notin Q, \mbox{ and both }  (a_-,b_{\pm}) \in Q \\
                         \ 0 \mbox{ otherwise}. \end{cases} 
\end{align}

For all of the cases considered in section \ref{sec:ex2DJAB-commuting} (i.e.~when $[A,B] = 0$), the value of $J_{AB}$ is determined exactly by its action on single points in the space $\mathbb{R}^2$, but this is no longer true when $[A,B] \neq 0$.  In particular, $J_{AB}(\{p\}) = 0$ for any $p \in \mathbb{R}^2$, and so $J_{AB}$ is clearly not a PVM in the non-commuting case considered here.  It is straightforward to see, however, that $J_{AB}$ \emph{is} a gPVM.  Also, from equation (\ref{eq:JAB-noncommuting}) it is easy to see that $J_{AB}$ has the correct margins.  Finally, since $\mathcal{H}$ is two-dimensional, all of the projectors which are in the image of the map $J_{AB}$ occur in the images of the PVMs $A$ and $B$ --- this is no longer generically true in three or higher dimensions, even when $A$ and $B$ commute. 
Although generalized joint observables on two-dimensional Hilbert spaces are relatively simple, they suffice to illustrate the differences between $J_{AB}$ in the commuting and non-commuting cases, as well as some of the basic features of generalized joint observables. 


\subsection{Coarse-graining and PVMs Associated with $J_{AB}$}
\label{sec:cgandgsfs}

Another interesting property of $J_{AB}$ is that it is well-behaved with regard to the procedure of coarse-graining.  In particular, given partitions $\mathcal{P}_A$ and $\mathcal{P}_B$ of $\mathbb{R}$ associated with coarse-grainings $\tilde{A}$ and $\tilde{B}$ of observables $A$ and $B$, respectively, there is a natural partition $\mathcal{P}_{AB}$ of $\mathbb{R}^2$ which allows us to define a coarse-graining $\tilde{J}_{AB}$ of $J_{AB}$ in the same manner as for PVMs.  It is straightforward to show that
\begin{equation}
\label{eq:constJ_AVcommCG}
\tilde{J}_{AB} = J_{\tilde{A} \tilde{B}}.
\end{equation}
That is, the construction of our $J_{AB}$ commutes with the operation of coarse-graining~\footnote{We actually require one technical condition on the partitions $\mathcal{P}_A$ and $\mathcal{P}_B$ in order for equation (\ref{eq:constJ_AVcommCG}) to hold --- see Theorem \ref{thm:coarse} in the appendix.}.  Moreover, when the coarse-graining is ``coarse enough'' (specifically, when there exists $\{ Q_1, Q_2, \ldots \} \subseteq \mathcal{P}_{AB}$ satisfying $\sum_{i} J_{AB}(Q_i) = I$), the coarse-grained joint observable $\tilde{J}_{AB}$ is in fact a PVM, not just a gPVM (see Theorem~\ref{thm:coarse-grain}).

There is another method by which we can construct a PVM from $J_{AB}$.  As we show in the appendix (Theorem \ref{thm:chain_works_J}), any gPVM $J$ on $(\Omega, \mathcal{M})$, along with a \emph{generating chain} $\mathcal{E}$ for $\mathcal{M}$ (i.e.~$\mathcal{E} \subseteq \mathcal{M}$ generates $\mathcal{M}$ as a Boolean $\sigma$-algebra, and the elements of $\mathcal{E}$ are totally ordered under inclusion), can be used to construct a unique PVM on $(\Omega, \mathcal{M})$ which agrees with $J$ on the elements of $\mathcal{E}$ \footnote{Just as the generating chain $\{ (-\infty, \lambda] \}_{\lambda \in \mathbb{R}}$ for $\mathcal{B}(\mathbb{R})$ can be used to define a spectral family $\{ E_{\lambda} \}_{\lambda \in \mathbb{R}}$ from which a PVM on $(\mathbb{R}, \mathcal{B}(\mathbb{R}))$ can be constructed, we can think of the generating chain $\mathcal{E}$ for $\mathcal{B}(\mathbb{R}^2)$ as giving rise to a ``generalized spectral family'' $\{ J(X) \}_{X \in \mathcal{E}}$ from which a PVM on $(\mathbb{R}^2, \mathcal{B}(\mathbb{R}^2))$ can be constructed.}.  In the case ${ (\Omega, \mathcal{M}) = (\mathbb{R}^2, \mathcal{B}(\mathbb{R}^2)) }$, $\mathcal{E}$ a generating chain for $\mathcal{B}(\mathbb{R}^2)$, and $J = J_{AB}$, we denote this PVM by $J_{AB}^{\mathcal{E}}$.  (Of course, since $J_{AB}^{\mathcal{E}}$ is a PVM, it is naturally associated (for \emph{any} partition of $\mathbb{R}^2$) with a trace preserving quantum operation.)  When $[A,B] = 0$, we have $J_{AB}^{\mathcal{E}} = J_{AB}$ for any generating chain $\mathcal{E}$ for $\mathcal{B}(\mathbb{R}^2)$.  When $[A,B] \neq 0$ this is not true, and no $J_{AB}^{\mathcal{E}}$ is a joint observable for $A$ and $B$ since in this case no joint observable exists.  The physical meaning of the $J^{\mathcal{E}}_{AB}$'s when $[A,B] \neq 0$ remains obscure.

\subsection{Other Characterizations of $J_{AB}$}
\label{sec:charJ_AB}

In addition to all of the aforementioned properties of $J_{AB}$, we have the following independent characterization of our generalized joint observable.  For any PVMs $A$ and $B$, and any set map $J: \mathcal{B}(\mathbb{R}^2) \rightarrow \mathcal{L}_{\mathcal{H}}$, we have that $J = J_{AB}$ if and only if $J$ satisfies the following two conditions for all $Q \in \mathcal{B}(\mathbb{R}^2)$ and all $\ket{\psi} \in \mathcal{H}$ (see Theorem \ref{thm:characterization} in the appendix):
\begin{enumerate}[(1)] 
\item If there exist $R_1,R_2 \in \mathcal{B}(\mathbb{R})$ with $R_1 \times R_2 \subseteq Q$ and $A(R_1) \wedge B(R_2) \ket{\psi} = \ket{\psi}$, then $J(Q) \ket{\psi} = \ket{\psi}$.
\item If for every $R_1,R_2 \in \mathcal{B}(\mathbb{R})$ with $R_1 \times R_2 \subseteq Q$ we have $A(R_1) \wedge B(R_2) \ket{\psi} = 0$, then $J(Q) \ket{\psi} = 0$.
\end{enumerate}
Qualitatively speaking, property (1) above states the following intuitive requirement on the generalized joint observable $J_{AB}$: for a given state $\ket{\psi}$, if the system has the value of $A$ in $R_1$ and the value of $B$ in $R_2$, and $R_1 \times R_2 \subseteq Q$, then the value of the generalized joint observable $J_{AB}$ is in the range $Q$.  Similarly, property (2) above states that if one can never (i.e.~with probability zero) measure the value of $A$ in $R_1$ and $B$ in $R_2$ for every $R_1 \times R_2 \subseteq Q$, then the value of the generalized joint observable $J_{AB}$ is never in the range $Q$.  

Alternatively, $J_{AB}$ has a characterization related to possible measurement outcomes.  Let $M_A^{\ket{\psi}}$ denote the set of possible measurement outcomes associated with $A$ when the system is in the state $\ket{\psi}$, i.e.
\begin{equation}
\label{mmtoutcomeset}
M_A^{\ket{\psi}} := \{ \lambda \in \mathbb{R} \ : \ A(\{\lambda\}) \ket{\psi} \neq 0 \} .
\end{equation}
Then, for any $Q \in \mathcal{B}(\mathbb{R}^2)$ for which $M_A^{\ket{\psi}}~\times~M_B^{\ket{\psi}} \subseteq Q$, we have that $J_{AB}(Q)\ket{\psi} = \ket{\psi}$, which is to say that if the system is in the state $\ket{\psi}$, a joint measurement of $A$ and $B$ is guaranteed to yield an outcome in $Q$.  While the reverse implication holds for rectangular sets $Q = R_1 \times R_2$, it does not hold in general. (That this is so follows from the fact that $\{ \ket{\psi} \in \mathcal{H} \ : \ M_A^{\ket{\psi}} \times M_B^{\ket{\psi}} \subseteq Q \}$ is not a subspace of $\mathcal{H}$ unless $Q$ is rectangular.)  Despite this, $J_{AB}(Q)$ is the projector onto the span of all the states $\ket{\psi} \in \mathcal{H}$ whose possible measurement outcomes associated with $A$ and $B$ are contained in $Q$ (i.e.~$M_A^{\ket{\psi}} \times M_B^{\ket{\psi}} \subseteq Q$).  So, another way of thinking of $J_{AB}$ is as the minimal (with respect to the partial order $\leq$ on $\mathcal{L}_{\mathcal{H}}$) set map from $\mathcal{B}(\mathbb{R}^2) \to \mathcal{L}_{\mathcal{H}}$ satisfying property (1) above.

\section{Functional Calculus}
\label{sec:funcs}

\subsection{Basic Properties}
\label{sec:bp-functionalcalculus}

Given any observables $A$ and $B$, along with a function ${f:\mathbb{R}^2 \to \mathbb{R}}$, we define the map $f(A,B):\mathcal{B}(\mathbb{R}) \to \mathcal{L}_{\mathcal{H}}$ by (for all $Q \in \mathcal{B}(\mathbb{R})$)
\begin{equation}
f(A,B)(Q) := \big( J_{AB} \circ f^{-1} \big) (Q),
\end{equation}
just as in the case $[A,B]=0$.  Using the fact that $J_{AB}$ is a gPVM, it is straightforward to show that $f(A,B)$ is also a gPVM (see Theorem \ref{thm:f_pvm} in the appendix).  Unlike polynomials in the ordinary linear algebraic sum and product, there are no ambiguities in defining $f(A,B)$ --- as an example, for the two-variable polynomials $p(x,y) = xy^2x$ and $q(x,y) = y x^2 y$ (which both represent the same function from $\mathbb{R}^2$ to $\mathbb{R}$), we do not generically have $A B^2 A = B A^2 B$ (where juxtaposition of operators denotes the usual linear algebraic product), but we \emph{do} have $p(A,B) = q(A,B)$.  Additionally, for any unitary operator $U$, the generalized observable $f(A,B)$ has the intuitive property 
\begin{equation}
f(U A U^\dag, U B U^\dag) = U f(A,B) U^\dag,
\end{equation} 
which follows directly from Lemma \ref{lem:unitary} in the appendix.  Finally, just as an unselected measurement of $J_{AB}$ corresponds to a non-trace preserving quantum operation, so too does an unselected measurement of $f(A,B)$.

Now, for any pair of observables $A$ and $B$, and any function $f:\mathbb{R}^2 \to \mathbb{R}$, it turns out that $f(A,B)$ has a family of PVMs associated with it (just as the gPVM $J_{AB}$ has an associated family of PVMs).  In particular, for each generating chain $\mathcal{E}$ of $\mathcal{B}(\mathbb{R})$, there exists a unique PVM $f_{\mathcal{E}}(A,B)$ agreeing with $f(A,B)$ on all $E \in \mathcal{E}$ (see Theorem \ref{thm:chain_works_J} and Corollary \ref{cor:chain_works} in the appendix).  Of course, when $[A,B] = 0$, $f(A,B)$ is a PVM, and ${ f_{\mathcal{E}}(A,B) = f(A,B) }$ for any generating chain $\mathcal{E}$.

As an example, consider arbitrary observables $A$ and $B$, along with the generating chain 
\begin{equation}
\label{eq:intervalsT}
\mathcal{E}^{\star} := \{ ( -\infty, \lambda ] \ : \ \lambda \in \mathbb{R} \},
\end{equation}
and define
\begin{equation}
\label{eq:spectral}
E_\lambda := J_{AB} \circ f^{-1} \big((- \infty, \lambda] \big).
\end{equation}
Then $\{  E_{\lambda} \}_{\lambda \in \mathbb{R}}$ is a spectral family of projectors on $\mathcal{H}$, which corresponds to an observable in the standard way, and this observable is $f_{\mathcal{E}^\star}(A,B)$.  Now, since $\mathcal{H}$ is finite dimensional,  $f_{\mathcal{E}^\star}(A,B)$ will have a finite number of distinct eigenvalues $\lambda_1 < \lambda_2 < \cdots < \lambda_n$.  As such, since $R_1 \times R_2 \subseteq f^{-1}((- \infty, \lambda])$ exactly when $f(a,b) \leq \lambda$ for every $a \in R_1$ and $b \in R_2$, it is straightforward to show that a state $\ket{\psi} \in \mathcal{H}$ is an eigenstate of $f_{\mathcal{E}^\star}(A,B)$ with eigenvalue $\lambda_i$ exactly when $\ket{\psi}$ is in the span of all states for which any possible measurement outcomes $a$ and $b$ of $A$ and $B$, respectively, satisfy $f(a,b) \leq \lambda_i$, but $\ket{\psi}$ is orthogonal to any state whose possible measurement outcomes $a'$ and $b'$ of $A$ and $B$, respectively, satisfy $f(a',b') \leq \lambda_{i-1}$ \footnote{This extra orthogonality condition is not necessary when $i = 1$ (and thus $\lambda_{i-1}$ does not exist).}.  We will return to $f_{\mathcal{E}^\star}(A,B)$, for some specific choices of $f$, shortly.

For any generating chain $\mathcal{E}$ for $\mathcal{B}(\mathbb{R})$, the observable $f_{\mathcal{E}}(A,B)$ satisfies the following nice properties (where, as before, $\sigma(A)$ denotes the spectrum of $A$): 
\begin{enumerate}[(1)] 
\item $\sigma \big( f_{\mathcal{E}}(A,B) \big) \subseteq f \big( \sigma(A) , \sigma(B) \big)$;
\item $f_{\mathcal{E}}(U A U^\dag, U B U^\dag) = U f_{\mathcal{E}}(A,B) U^\dag$ for any unitary operator $U$ on $\mathcal{H}$;
\item If $A \ket{\psi} = a \ket{\psi}$ and $B \ket{\psi} = b \ket{\psi}$ for some $\ket{\psi} \in \mathcal{H}$, then $f_{\mathcal{E}}(A,B) \ket{\psi} = f(a,b) \ket{\psi}$.
\end{enumerate}
These are proved in the appendix (Theorem \ref{thm:f_nice}).  Additionally, it follows from property (1) that for any subset $S$ of $\mathbb{R}$ which contains all eigenvalues of both $A$ and $B$, and for any $f:\mathbb{R}^2 \rightarrow \mathbb{R}$ such that $f(x,y) \in S$ for all $x,y \in S$, we have that all eigenvalues of $f_{\mathcal{E}}(A,B)$ are also in $S$.  For example, if the eigenvalues of both $A$ and $B$ are positive integers and $f:\mathbb{R}^2 \rightarrow \mathbb{R}$ is addition, then all of the eigenvalues of the ``sum'' $f_{\mathcal{E}}(A,B)$ will also be positive integers.  


We now present an equality which illustrates the naturality of our functional calculus.  We begin with some definitions.  Let $f,g:\mathbb{R}^2 \to \mathbb{R}$ be addition and multiplication of numbers respectively, and let $e:\mathbb{R} \to \mathbb{R}$ be the exponential function.  For any observables $A$ and $B$, and for any fixed generating chain $\mathcal{E}$ of $\mathcal{B}(\mathbb{R})$, let $A \stackrel{\LARGE \textbf{.}}{+} B := f_{\mathcal{E}}(A,B)$ and $A \stackrel{\LARGE \textbf{.}}{\times} B := g_{\mathcal{E}}(A,B)$. If $e^{-1}(E) \in \mathcal{E}$ for any $E \in \mathcal{E}$ (as is the case for $\mathcal{E}^{\star}$ in (\ref{eq:intervalsT})), then
\begin{equation}
\label{genBCH}
 e^A \stackrel{\LARGE \textbf{.}}{\times} e^B  = e^{A \stackrel{\LARGE \textbf{.}}{+} B} .
\end{equation}
This statement follows directly from Lemmas \ref{lem:composition} and \ref{lem:right_comp} in the appendix.  
Comparing (\ref{genBCH}) above to the Baker-Campbell-Hausdorff formula involving the ordinary linear algebraic sum and product of non-commuting observables
\begin{equation}
e^A e^B = e^{A + B + \frac{1}{2}[A,B] + \frac{1}{12}[A,[A,B]] - \frac{1}{12}[B,[A,B]] + \cdots },
\end{equation}
one can clearly see the elegance and simplicity of the new functional calculus.  Our addition and multiplication also satisfy other nice properties, such as commutativity
\begin{equation}
A \stackrel{\textbf{.}}{\times} B = B \stackrel{\textbf{.}}{\times} A \quad \textrm{and} \quad A \stackrel{\textbf{.}}{+} B = B \stackrel{\textbf{.}}{+} A
\end{equation}
for any $A$ and $B$ (even when $[A,B] \neq 0$) and any $\mathcal{E}$.  However, in some ways the behavior is not as natural --- for example, our $\stackrel{\textbf{.}}{+}$ and $\stackrel{\textbf{.}}{\times}$ are not, in general, associative, and generically $\stackrel{\textbf{.}}{\times}$ does not distribute over $\stackrel{\textbf{.}}{+}$.  

Finally, we note another natural property of our functional calculus.  For observables $A$ and $B$ we write $A \sqsubseteq B$ if 
\begin{equation}
\label{spectralorder}
A((- \infty, \lambda]) \leq B((- \infty, \lambda]) \quad \forall \lambda \in \mathbb{R} ,
\end{equation}
which is referred to as the \emph{spectral order} (see e.g., \cite{olson_71}) and differs from the usual order on Hermitian operators (which is defined by $A \leq B$ if ${\bra{\psi} A \ket{\psi} \leq \bra{\psi} B \ket{\psi}}$ for all ${\ket{\psi} \in \mathcal{H}}$).  In general, we have that if ${ A \sqsubseteq B }$, then ${ A \leq B } $, but not conversely; however, these orderings agree when $A$ and $B$ are projection operators, as well as when $A$ and $B$ commute.
Now, for $\stackrel{\textbf{.}}{+}$ defined as above and with respect to the generating chain $\mathcal{E}^{\star}$ in (\ref{eq:intervalsT}), if $A$ and $B$ are observables such that $A \sqsubseteq B$, then (for any observable~$C$)  
\begin{equation}
\label{spectralorder}
A \stackrel{\textbf{.}}{+} C \sqsubseteq B \stackrel{\textbf{.}}{+} C,
\end{equation}
i.e.~the addition of observables defined with respect to $\mathcal{E}^{\star}$ respects the spectral order on the observables. 
(By contrast, the ordinary linear algebraic sum does not respect the spectral order \cite{olson_71}.)  A similar statement holds for the operation $\stackrel{\textbf{.}}{\times}$ defined with respect to $\mathcal{E}^{\star}$ when the eigenvalues of the observables involved are non-negative numbers.

 We now present a simple worked-out example of our functional calculus.

\subsection{Example: $f(A,B)$ when $\dim \mathcal{H} = 2$}
\label{sec:examples}

In what follows, we take $\mathcal{H}$ to be a two-dimensional Hilbert space, and $f:\mathbb{R}^2 \to \mathbb{R}$ to be addition.  While this example is somewhat trivial, it illustrates some interesting points.
As in section \ref{sec:ex2DJAB-commuting}, we take $A= \alpha I + \mathbf{a} \cdot \boldsymbol{\sigma}$ and $B = \beta I + \mathbf{b} \cdot \boldsymbol{\sigma}$, where the eigenvalues are given by $a_{\pm} := \alpha \pm | \mathbf{a}|$ and $b_{\pm} := \beta \pm | \mathbf{b} |$, respectively.  Without loss of generality, we assume that $|\mathbf{a} | \geq | \mathbf{b} |$. 

First, recall that when $[A,B]=0$, we have that $f(A,B) = A+B$.  However, when $[A,B] \neq 0$, then $f(A,B)$ is only a gPVM (i.e.~it is not a PVM).  Of course, in this case $A$ and $B$ must each have two distinct eigenvalues.  One can clearly see that $f(A,B)$ is not a PVM by computing
\begin{align}
f(A,B) \big( (-\infty, a_- + b_- ] \big) & = 0 \nonumber \\
f(A,B) \big( (a_- + b_- , a_+ + b_+ ) \big) & = 0 \nonumber \\
f(A,B) \big( [a_+ + b_+, \infty) \big) & =0 ,   
\end{align}
and noting that if $f(A,B)$ were a PVM, the above three terms would need to sum to $I$.

As discussed above, we can form an observable ${ A \stackrel{\LARGE \textbf{.}}{+} B := f_{\mathcal{E}}(A,B) }$ by choosing a generating chain $\mathcal{E}$ for $\mathcal{B}(\mathbb{R})$.  A natural choice is to take $\mathcal{E} = \mathcal{E}^{\star}$ from (\ref{eq:intervalsT}), in which case the spectral family $\{ E^{A \stackrel{\LARGE \textbf{.}}{+} B}_{\lambda} \}_{\lambda \in \mathbb{R}}$ for $A \stackrel{\LARGE \textbf{.}}{+} B$ is given by 
\begin{equation}
\label{eq:plus}
E^{A \stackrel{\LARGE \textbf{.}}{+} B}_\lambda = J_{AB} \circ f^{-1} \big((- \infty, \lambda] \big) = \bigvee_{\eta \in \mathbb{R}} E^A_\eta \wedge E^B_{\lambda - \eta} ,
\end{equation}
where $\{ E_{\lambda}^A \}_{\lambda \in \mathbb{R}}$ and $\{ E_{\lambda}^B \}_{\lambda \in \mathbb{R}}$ are the spectral families for $A$ and $B$, respectively.

A straightforward calculation then yields (whenever $[A,B] \neq 0$)
\begin{equation}
A \stackrel{\LARGE \textbf{.}}{+} B = (\alpha + \beta) I + (| \mathbf{a}| - |\mathbf{b}|) \hat{\mathbf{a}} \cdot \boldsymbol{\sigma},
\end{equation}
where $\hat{\mathbf{a}} = \mathbf{a}/|\mathbf{a}|$.  Note that the eigenvalues of the observable $A  \stackrel{\LARGE \textbf{.}}{+} B$ are $a_+ + b_-$ and $a_- + b_+$.  Compare this to the eigenvalues of the linear algebraic sum $A + B$, which are given by $(\alpha + \beta) \pm | \mathbf{a} + \mathbf{b}|$, and are clearly not of the above form.

We can easily see from the above results that $\stackrel{.}{+}$ is not an associative operation.  Let $\alpha = \beta = 0$ and let $C = \mathbf{c} \cdot \boldsymbol{\sigma}$.  Then the eigenvalues of $(A \stackrel{\LARGE \textbf{.}}{+} B) \stackrel{\LARGE \textbf{.}}{+} C$ are 
\begin{equation}
\pm ( \big| |\mathbf{a}| - | \mathbf{b} | \big| - |\mathbf{c}| )
\end{equation}
while the eigenvalues of $A \stackrel{\LARGE \textbf{.}}{+} (B \stackrel{\LARGE \textbf{.}}{+} C)$ are
\begin{equation}
\pm (  |\mathbf{a}| - \big| | \mathbf{b} |  - |\mathbf{c}| \big| ).
\end{equation} 

The simple example above explicitly illustrates some of the general features of our functional calculus discussed in section \ref{sec:bp-functionalcalculus}.  

\vspace{15pt}

\section{Conclusion}
\label{sec:conclusion}

In the context of finite-dimensional Hilbert spaces, we have given an explicit construction of a generalized joint observable $J_{AB}$ for an arbitrary pair of observables $A$ and $B$, as well as described a realization of the corresponding joint measurement in the framework of quantum operations, both of which agree with their standard definitions when $[A,B] = 0$.  Further, we have noted that the failure of generalized joint observables to be ordinary joint observables is characterized by their lack of countable additivity (as they are only countably sub-additive), and have demonstrated how this failure can be interpreted as a manifestation of the uncertainty principle.  
We then went on to describe the functional calculus of observables which arises from our notion of a generalized joint observable, as well as described some of its remarkable properties.   

Although the results presented in sections \ref{sec:generalizedJOs} and \ref{sec:funcs} are for observables on finite-dimensional Hilbert spaces, the appendix extends these results to bounded observables with pure point spectra on any separable complex Hilbert space.  Also, although we have chosen (for notational simplicity) to present our results for \emph{pairs} of observables, they can all be extended in a straightforward fashion to sets of $n$ observables. In particular, equation (\ref{J_AB}) can be generalized to (for observables $A_1, \ldots, A_n$ and $Q \in \mathcal{B}(\mathbb{R}^n)$) 
\begin{equation}
J_{A_1, \ldots, A_n}(Q) := \hspace{-20pt} \bigvee_{\substack{R_1 \times \ldots \times R_n \subseteq Q \\ R_1 \times \ldots \times R_n \in \mathcal{B}(\mathbb{R}^n)}} \hspace{-20pt} \big[ A_1(R_1) \wedge \ldots \wedge A_n(R_n)\big],
\end{equation}
and one obtains the functional calculus in this case for (Borel measurable) functions $f:\mathbb{R}^n \to \mathbb{R}$ by defining $f(A_1, \ldots, A_n) := J_{A_1, \ldots, A_n} \circ f^{-1}$, which still has all of the interesting properties discussed in section \ref{sec:bp-functionalcalculus}.

This work opens up many directions for further research, perhaps the most pressing of which is to find interesting problems which are more naturally formulated in terms of our functional calculus of observables instead of the ordinary linear algebraic sum and product, and for which the new calculus provides novel physical insight.  Additionally, there are questions which are technical in nature that we would like to address --- e.g.~we would like to extend the notions of generalized joint observables and joint measurability presented here to observables with continuous portions to their spectra, as well as further explore properties of the families of PVMs associated with gPVMs.  Finally, it would be interesting to design simple experiments in which non-commuting observables are simultaneously measured in the manner outlined in section \ref{sec:JOs&JM}.   


\begin{acknowledgments}

We thank Randall Espinoza, Nick Huggett, Mark Mueller, and Josh Norton for many useful discussions.







\end{acknowledgments}



\section{Technical Appendix}

In this section we prove all of our previous results; moreover, we do this in a more general context than stated originally.  In what follows $\mathcal{H}$ will denote a fixed separable complex Hilbert space with projection lattice $\mathcal{L}_{\mathcal{H}}$.  Additionally, we assume the reader is conversant with the standard terminology and basic results used in functional analysis (see e.g., \cite{Kreyszig}).  Finally, we will take $\mathbb{N} = \{1,2,\ldots \}$ to denote the positive integers.

\subsection{Basic Definitions and Properties}

\begin{kap}
Let $(\Omega, \mathcal{M})$ be a measurable space, and let $A:\mathcal{M} \to \mathcal{L}_{\mathcal{H}}$ be such that
\begin{enumerate}[(1)]
\item $A(\Omega) = I$;
\item $A(R) \perp A(S)$ for all disjoint $R,S \in \mathcal{M}$;
\item $\displaystyle \sum_{i=1}^\infty A(R_i) = A \Big( \bigcup_{i=1}^\infty R_i \Big)  $ for all $R_1, R_2, \ldots \in \mathcal{M}$  such that $R_i \cap R_j = \emptyset$ whenever $i \neq j$. 
\end{enumerate}
Then $A$ is called a \emph{PVM} on $(\Omega, \mathcal{M})$, and if furthermore $\Omega = \mathbb{R}$ and $\mathcal{M} = \mathcal{B}(\mathbb{R})$ (the Boolean $\sigma$-algebra of Borel subsets of $\mathbb{R}$), then we call $A$ a \emph{standard PVM}. 
\end{kap}

The standard PVMs are in 1-1 correspondence with (not necessarily bounded) self-adjoint operators on $\mathcal{H}$.  
Using the spectral family $\{ E_{\lambda} \}_{\lambda \in \mathbb{R}}$ defined by a standard PVM (i.e.~${ A((- \infty, \lambda]) = E_{\lambda} }$ for all ${ \lambda \in \mathbb{R} }$), the corresponding self-adjoint operator $A$ on $\mathcal{H}$ is obtained by the Riemann-Stieltjes integral   
\begin{equation}
A = \int_{- \infty}^{+ \infty} \lambda d E_{\lambda} ,
\end{equation} 
which is defined to converge in the strong operator topology.  Conversely, given a self-adjoint operator whose (right continuous) spectral family is denoted by $\{ E_{\lambda} \}_{\lambda \in \mathbb{R}}$, the corresponding standard PVM $A: \mathcal{B}(\mathbb{R}) \rightarrow \mathcal{L}_{\mathcal{H}}$ is defined to be the unique PVM which satisfies (for all $\alpha, \beta \in \mathbb{R}$ with $\alpha \leq \beta$)
\begin{equation}
A((\alpha, \beta]) = \int_{\alpha}^{\beta} d E_{\lambda}.
\end{equation} 

\begin{kap}
Let $A$ be a PVM.  Then we say $A$ is \emph{diagonalizable} if $A$ is standard and the self-adjoint operator corresponding to $A$ is bounded and has a set of eigenvectors which forms a (Schauder) basis for $\mathcal{H}$.
\end{kap}
Note that $A$ diagonalizable is equivalent to the self-adjoint operator corresponding to $A$ being bounded with pure point spectrum.  Moreover, for $A$ diagonalizable, we define $\sigma_p(A)$ to be the set of eigenvalues of the self-adjoint operator corresponding to $A$, and $\sigma(A)$ to be the spectrum of this operator.  Then we have that (i) $\sigma(A)$ is compact, (ii) $\sigma_p(A)$ is countable, (iii) $\sigma(A)$ is the closure of $\sigma_p(A)$, and finally (iv) ${ A(R) = A(R \cap \sigma(A)) = A(R \cap \sigma_p(A)) }$ for any ${ R \in \mathcal{B}(\mathbb{R}) }$ \cite{Kreyszig}.

As mentioned previously, our generalized joint observables as defined in section \ref{sec:generalizedJOs} above are \emph{not} quite PVMs --- instead, they are generalized projection-valued measures (or gPVMs) as defined below.  We note that the properties of gPVMs actually make them analogous to the notion of an \emph{inner measure}, rather than a typical measure \cite{halmos}.


\begin{kap}
Let $(\Omega,\mathcal{M})$ be a measurable space.  A \emph{generalized projection-valued measure}, or \emph{gPVM}, on ($\Omega, \mathcal{M}$) (or just $\Omega$, if $\mathcal{M}$ is clear from the context) is a map $J:\mathcal{M} \to \mathcal{L}_{\mathcal{H}}$ such that
\begin{enumerate}[(1)]
\item $J(\Omega) = I$;
\item $J(R) \perp J(S)$ for all disjoint $R,S \in \mathcal{M}$;
\item $\displaystyle \sum_{i=1}^\infty J(R_i) \leq J \Big( \bigcup_{i=1}^\infty R_i \Big)$ for all $R_1, R_2, \ldots \in \mathcal{M}$ such that $R_i \cap R_j = \emptyset$ whenever $i \neq j$.  
\end{enumerate}   
\end{kap}  

Note that every PVM is also trivially a gPVM.  We now prove some useful properties of gPVMs.

\begin{lemmy}
\label{lem:gpvm_props}
Let $(\Omega, \mathcal{M})$ be a measurable space, and let $J$ be a gPVM on $(\Omega, \mathcal{M})$.  Then
\begin{enumerate}[(1)]
\item $J(\emptyset) = 0$;
\item $J(R) \leq J(S)$ whenever $R, S \in \mathcal{M}$ with $R \subseteq S$;
\item $\displaystyle \bigvee_{i=1}^n J(R_i) = J(R_n)$ and $\displaystyle \bigwedge_{i=1}^n J(R_i) = J(R_1)$ for $R_1, \ldots, R_n \in \mathcal{M}$ with $R_1 \subseteq \cdots \subseteq R_n$;
\item $\displaystyle \bigvee_{i=1}^\infty J(R_i) \leq J \Big( \bigcup_{i=1}^\infty R_i \Big)$ and $\displaystyle \bigwedge_{i=1}^\infty J(R_i) \geq J \Big( \bigcap_{i=1}^\infty R_i \Big)$ for any $R_1, R_2, \ldots \in \mathcal{M}$.
\end{enumerate}
\end{lemmy} 
\begin{proof} 
Regarding property 1, we clearly have ${ \Omega \cap \emptyset = \emptyset }$, so property 2 of gPVMs gives that ${ J(\emptyset) \bot J(\Omega) = I }$, from which it follows that ${J(\emptyset) = 0 }$.   

Next, if $R \subseteq S$, then $R \cap (S - R) = \emptyset$, so that by sub-additivity (i.e.~property 3 of gPVMs) we have
\begin{equation}
J(R) \leq J(R) + J(S-R) \leq J(R \cup (S-R)) = J(S),
\end{equation} 
which shows that property 2 holds.

To see that property 3 holds, let $R_i \in \mathcal{M}$ for $i \in \{ 1, \ldots, n \}$ be such that ${ R_1 \subseteq \ldots \subseteq R_n }$, and note that $J(R_i) \leq \bigvee_{i=1}^{n} J(R_i)$ for each $i \in \{ 1, \ldots, n \}$ --- in particular, $J(R_n) \leq \bigvee_{i=1}^{n} J(R_i)$.  Now, since $R_i \subseteq \bigcup_{i = 1}^{n} R_i$ for each $i \in \{ 1, \ldots, n \}$, property 2 above gives that ${ J(R_i) \leq J(\bigcup_{i = 1}^{n} R_i) = J(R_n) }$ for each $i \in \{ 1, \ldots, n \}$, so that ${ \bigvee_{i=1}^{n} J(R_i) \leq J(\bigcup_{i = 1}^{n} R_i) = J(R_n) }$.  This inequality, along with that above, establishes the first equality in property 3. 
The other expression in property 3 above is obtained in a similar fashion. 

Finally, to see that property 4 holds, let $R_i \in \mathcal{M}$ for $i \in \mathbb{N}$.  Since $R_i \subseteq \bigcup_{i = 1}^{\infty} R_i$ for each $i \in \mathbb{N}$, property 2 above gives that $J(R_i) \leq J(\bigcup_{i = 1}^{\infty} R_i)$ for each $i \in \mathbb{N}$.  Thus, 
\begin{equation}
\bigvee_{i=1}^{\infty} J(R_i) \leq J(\bigcup_{i = 1}^{\infty} R_i).
\end{equation}
The other expression in property 4 above is obtained in a similar fashion.

\end{proof}

The following characterization of gPVMs will also prove useful.

\begin{lemmy}
\label{lem:char}
Let $(\Omega, \mathcal{M})$ be a measurable space, and let $J:\mathcal{M} \to L_{\mathcal{H}}$ satisfy
\begin{enumerate}[(1)]
\item $J(\Omega) = I$;
\item $J(R) \perp J(S)$ for all disjoint $R,S \in \mathcal{M}$.
\end{enumerate}
Then $J$ is a gPVM iff 
\begin{equation}
\label{eq:property3}
J(R) \leq J(S) \quad \textrm{for all} \quad R,S \in \mathcal{M} \quad \textrm{with} \quad R \subseteq S. 
\end{equation}
\end{lemmy}
\begin{proof}
If $J$ is a gPVM, then by (2) in lemma \ref{lem:gpvm_props}, equation (\ref{eq:property3}) is satisfied.  Conversely, if we have $R_1, R_2, \ldots \in \mathcal{M}$ pairwise disjoint, then $J(R_i) \leq J \big( \bigcup_{i=1}^\infty R_i \big)$ for all $i \in \mathbb{N}$, since $J$ satisfies equation (\ref{eq:property3}), and hence we have
\begin{equation}
\sum_{i=1}^\infty J(R_i) = \bigvee_{i=1}^\infty J(R_i) \leq J \Big( \bigcup_{1=1}^\infty R_i \Big)
\end{equation}
where the first equality is due to the fact that the $R_i$'s are pairwise disjoint.
\end{proof}

Given two gPVMs $A$ and $B$ on a measurable space $(\Omega,\mathcal{M})$, we define their joint observable $J_{AB}:\mathcal{M}^2 \to \mathcal{L}_{\mathcal{H}}$ by 
\begin{equation}
\label{def:joint_dist}
J_{AB}(Q) := \hspace{-10pt} \bigvee_{\substack{R_1 \times R_2 \subseteq Q \\ R_1,R_2 \in \mathcal{M}}} \hspace{-10pt} A(R_1) \wedge B(R_2) \quad \forall Q \in \mathcal{M} , 
\end{equation}
where $\mathcal{M}^2$ denotes the product $\sigma$-algebra of $\mathcal{M}$ with itself, i.e.~the smallest $\sigma$-algebra over $\Omega^2$ which contains the Cartesian product $\mathcal{M} \times \mathcal{M}$.
As can easily be seen in the case where $A,B$ are PVMs corresponding to non-commuting observables, $J_{AB}$ so defined is not in general a PVM, but only a gPVM.

\begin{thrm}
\label{thm:joint_pvim}
Given any two gPVMs $A$ and $B$ on a measurable space ($\Omega,\mathcal{M}$), $J_{AB}$ defined above is a gPVM on $\Omega^2$.
\end{thrm}
\begin{proof}
First, we have that 
\begin{equation}
J_{AB}(\Omega^2) = \hspace{-15pt} \bigvee_{\substack{R_1 \times R_2 \subseteq \Omega^2 \\ R_1,R_2 \in \mathcal{M}}} \hspace{-15pt} A(R_1) \wedge B(R_2) \geq A(\Omega) \wedge B(\Omega) = I.
\end{equation}
Next, we note that if $Q \cap R = \emptyset$, then clearly $Q_0 \cap R_0 = \emptyset$ for all $Q_0 \subseteq Q$ and $R_0 \subseteq R$.  Hence we have
\begin{equation}
\label{eq:perp}
A(Q_1) \wedge B(Q_2) \perp A(R_1) \wedge B(R_2)
\end{equation}
for all $Q_1 \times Q_2 \subseteq Q$ and $R_1 \times R_2 \subseteq R$.  Taking the join over each side of expression (\ref{eq:perp}) gives $J_{AB}(Q) \perp J_{AB}(R)$.  We also have that equation (\ref{eq:property3}) is satisfied since when $Q \subseteq R$, we have that any $R_1 \times R_2 \subseteq Q$ also satisfies $R_1 \times R_2 \subseteq R$.  Hence $J_{AB}$ is a gPVM by Lemma \ref{lem:char}.
\end{proof}

\subsection{Uncertainty Relation}

We begin by extending our definition of $M_A^{\ket{\psi}}$ (equation (\ref{mmtoutcomeset}) in section \ref{sec:charJ_AB}) for a given diagonalizable PVM $A$ and state $\ket{\psi} \in \mathcal{H}$ --- namely 
\begin{equation}
M_A^{\ket{\psi}} := \{ \lambda \in \mathbb{R} \ : \ A(\{\lambda\}) \ket{\psi} \neq 0 \} .
\end{equation}
Note that this agrees with our previous definition for finite-dimensional $\mathcal{H}$.  This allows us to state and prove rigorously a generalization of our earlier result (i.e.~expression (\ref{ourUncertaintyREL})) concerning the uncertainty relation.

\begin{thrm}
\label{thm:uncertaintyREL}
Let $A,B$ be diagonalizable PVMs, let ${ \ket{\psi} \in \mathcal{H} }$ with ${ \langle \psi \ket{\psi} = 1 }$, and let ${ L_A,L_B \subseteq \mathbb{R} }$ be intervals with lengths $l_A$ and $l_B$, respectively, such that ${ J_{AB}(L_A \times L_B) \ket{\psi} = \ket{\psi} }$.  Then, 
\begin{equation}
l_Al_B \geq 2 \Delta A \Delta B \geq |\langle [A,B] \rangle | .
\end{equation}   
\end{thrm}
\begin{proof}
First, we have 
$M_A^{\ket{\psi}} \subseteq L_A$ and $M_B^{\ket{\psi}} \subseteq L_B$, since $J_{AB}(L_A \times L_B) = A(L_A) \wedge B(L_B)$.  Then since $A$ and $B$ are diagonalizable, we can expand
\begin{equation}
\ket{\psi} = \sum_i \alpha_i \ket{a_i} = \sum_j \beta_j \ket{b_j}
\end{equation}
where $A \ket{a_i} = a_i \ket{a_i}$ and $B \ket{b_j} = b_j \ket{b_j}$.  Then 
\begin{equation} 
\langle A \rangle^2 = \Big( \sum_i |\alpha_i|^2 a_i \Big)^2 = \sum_{ij} | \alpha_i |^2 |\alpha_j|^2 a_i a_j
\end{equation}
and 
\begin{equation}
\langle A^2 \rangle = \sum_i |\alpha_i|^2 a_i^2 = \sum_{ij} |\alpha_i|^2 |\alpha_j|^2 a_i^2,
\end{equation}
since $\sum_i |\alpha_i|^2 = 1$.  Hence, we have that
\begin{align}
2 (\Delta A)^2 & = \sum_{ij} |\alpha_i|^2 |\alpha_j|^2 (a_i^2 + a_j^2 - 2a_i a_j) \nonumber \\
& = \sum_{ij} |\alpha_i|^2 |\alpha_j|^2 |a_i - a_j|^2 \nonumber \\
& \leq | \sup M_A^{\ket{\psi}} - \inf M_A^{\ket{\psi}}|^2 \leq l_A^2 ,
\end{align}
since $\alpha_i = 0$ unless $a_i \in  M_A^{\ket{\psi}}$.  Similarly, we have that $2 (\Delta B)^2 \leq l_B^2$, and putting these results together yields $l_Al_B \geq 2 \Delta A \Delta B \geq | \langle [A,B] \rangle |$, where the second inequality is just the Robertson uncertainty relation.
\end{proof}

\subsection{Coarse-graining}


\begin{lemmy}
\label{lem:partition_sums}
Let $J$ be a gPVM on the measurable space $(\Omega,\mathcal{M})$, and let $E_i \in \mathcal{M}$ (for each $i \in \mathbb{N}$), satisfy $E_i \cap E_j = \emptyset$ whenever $i \neq j$, and also
\begin{equation}
\label{eq:is_I-0}
\bigvee_{i=1}^\infty J(E_i) = I.
\end{equation}
Then for any subset $S \subseteq \mathbb{N}$, we have that 
\begin{equation}
\label{eq:partition_sums}
\bigvee_{i \in S} J(E_i) = J \Big( \bigcup_{i \in S} E_i \Big).
\end{equation}
\end{lemmy}
\begin{proof}

We prove this by contradiction, so assume that equation (\ref{eq:partition_sums}) does not hold.  Since $J$ is a gPVM (so that, in particular, property 2 of Lemma \ref{lem:gpvm_props} holds), this means that
\begin{equation}
\label{eq:partition_less}
\bigvee_{i \in S} J(E_i) < J \Big( \bigcup_{i \in S} E_i \Big),
\end{equation}
and hence there must be some non-zero $\ket{\psi} \in \mathcal{H}$ with $J\big( \bigcup_{i \in S} E_i \big) \ket{\psi} = \ket{\psi}$ but $\bigvee_{i \in S} J(E_i) \ket{\psi} = 0$.  This means that $J(E_i) \ket{\psi} = 0$ for each $i \in S$.  But then, by equation (\ref{eq:is_I-0}) and the fact that the $E_i$'s are disjoint (so that the least upper bound is just the sum), we must have that 
\begin{equation}
\bigvee_{i \in S^c} J(E_i) \ket{\psi} = \ket{\psi},
\end{equation}
and hence, since $J$ is a gPVM, we have (by property 2 in Lemma \ref{lem:gpvm_props}) that 
\begin{equation}
J \Big( \bigcup_{i \in S^c} E_i \Big) \ket{\psi} = \ket{\psi}.
\end{equation}
But since $J$ is a gPVM, this leads to a contradiction, since  
\begin{equation}
J \Big( \bigcup_{i \in S^c} E_i \Big) \perp J \Big( \bigcup_{i \in S} E_i \Big)  
\end{equation}
implies that $\langle \psi \ket{\psi} = 0$, contradicting the fact that $\ket{\psi}$ was non-zero.
\end{proof}

\begin{lemmy}
\label{lem:partition_sets_sums}
Let $J$ be a gPVM on the measurable space $(\Omega,\mathcal{M})$, and let $E_i \in \mathcal{M}$ (for each $i \in \mathbb{N}$), satisfy $E_i \cap E_j = \emptyset$ whenever $i \neq j$, and also
\begin{equation}
\label{eq:is_I-1}
\bigvee_{i=1}^\infty J(E_i) = I.
\end{equation}
Then for any collection $\{ Q_i \}_{i=1}^{\infty}$ with $Q_i \in \mathcal{M}$ satisfying either 
\begin{equation}
\label{eq:q_prop}
Q_i \cap E_j = E_j \ \textrm{or}  \ Q_i \cap E_j = \emptyset
\end{equation}
for each $i,j \in \mathbb{N}$, we have that
\begin{equation}
\label{eq:sets_sum}
\bigvee_{i=1}^\infty J(Q_i) = J \Big( \bigcup_{i=1}^\infty Q_i \Big).
\end{equation}
\end{lemmy}
\begin{proof}
First, since $J$ is a gPVM, by property 2 in Lemma \ref{lem:gpvm_props} we have 
\begin{equation}
\label{eq:ineq1}
\bigvee_{i=1}^\infty J(Q_i) \leq J \Big( \bigcup_{i=1}^\infty Q_i \Big).
\end{equation}
It remains to show the other inequality.  Define 
\begin{equation}
E_0 = \Big( \bigcup_{i=1}^\infty E_i \Big)^c,
\end{equation}
and note that $E_0 \in \mathcal{M}$, as well as that $\{E_i\}_{i=0}^\infty$ is a partition of $\Omega$, and also that $J(E_0) = 0$.  Next, for $j \in \mathbb{N}$, define 
\begin{equation}
\mathcal{E}_j := \{ E_i \, : \, E_i \subseteq Q_j, \ i \in \mathbb{N} \} \cup \{ E_0\},
\end{equation}
and note that $Q_j \subseteq \bigcup \mathcal{E}_j$ by equation (\ref{eq:q_prop}), as well as that (recalling that $\bigcup X := \bigcup_{Z \in X} Z $ for any set $X$ whose elements $Z$ are, themselves, sets) 
\begin{align}
\bigvee_{E \in \mathcal{E}_j} J(E) &  =  \bigvee_{ \substack{ E \in \mathcal{E}_j \\ E \neq E_0 }} J(E) \leq J(Q_j)   \nonumber \\ 
& \qquad \leq J \Big( \bigcup \mathcal{E}_j \Big) = \bigvee_{E \in \mathcal{E}_j} J(E).
\end{align}
The first equality in the above expression holds since $J(E_0) = 0$, and the following inequalities follow from the fact that if $E_i \in \mathcal{E}_j$ with $i \neq 0$, then $E_i \subseteq Q_j$ and also that $Q_j \subseteq \bigcup \mathcal{E}_j$ (along with the fact that $J$ is a gPVM); the final equality then follows from Lemma \ref{lem:partition_sums}.  Hence, we have
\begin{equation}
J(Q_j) = \bigvee_{E \in \mathcal{E}_j} J(E). 
\end{equation}
Further, define $\mathcal{E} = \bigcup_{j=1}^\infty \mathcal{E}_j$.  Then (again since $J$ is a gPVM), we have
\begin{align}
J \Big( \bigcup_{j=1}^\infty Q_i \Big) & \leq  J \Big( \bigcup \mathcal{E} \Big) = \bigvee_{E \in \mathcal{E}} J(E) \nonumber \\
& \quad = \bigvee_{j=1}^\infty \Big( \bigvee_{E \in \mathcal{E}_j} J(E) \Big) = \bigvee_{j=1}^\infty J(Q_j), \label{eq:ineq2} 
\end{align}
where we have again used Lemma \ref{lem:partition_sums} to arrive at the second equality.  The desired result then follows from the inequalities (\ref{eq:ineq1}) and (\ref{eq:ineq2}).
\end{proof}

\begin{kap}
\label{kap:CGmeasurablespace}
Let $(\Omega,\mathcal{M})$ be a measurable space, and let $\mathcal{P}$ be a partition of $\Omega$.  If $\mathcal{P} \subseteq \mathcal{M}$, we define
\begin{equation}
\label{eq:part_measure}
\mathcal{M}_{\mathcal{P}} := \{ Q \subseteq \mathcal{P} \ : \ \bigcup Q \in \mathcal{M} \}.
\end{equation}
\end{kap}

It is easy to see that $\mathcal{M}_{\mathcal{P}}$ as defined in (\ref{eq:part_measure}) is indeed a (Boolean) $\sigma$-algebra.  We can then naturally define a coarse-graining of any gPVM with respect to this partition --- namely define, for any $\tilde{Q} \in \mathcal{M}_P$,  
\begin{equation}
\label{eq:coarse}
\tilde{J} (\tilde{Q}) := J \Big( \bigcup \tilde{Q} \Big) .
\end{equation}
We then find that our gPVMs behave naturally with respect to this notion of coarse-graining, provided the coarse-graining is ``well-behaved'' with respect to the measurable sets, as illustrated in the following two theorems.

\begin{thrm}
\label{thm:coarse-grain}
Let $J$ be a gPVM on the measurable space $(\Omega, \mathcal{M})$, and let $\mathcal{P} \subseteq \mathcal{M}$ be a partition of $\Omega$.  Then $\tilde{J}$ is a gPVM on the measurable space $(\mathcal{P}, \mathcal{M}_P)$.  If, furthermore, there is some countable subset $\mathcal{P}_0 \subseteq \mathcal{P}$ such that
\begin{equation}
\label{eq:sum_to_I}
\sum_{E \in \mathcal{P}_0} J(E) = I,
\end{equation}
then $\tilde{J}$ is a PVM.
\end{thrm}
\begin{proof}
First we show that $\tilde{J}$ is a gPVM whenever $\mathcal{P}$ is a partition of $\Omega$.  Clearly, $\tilde{J}$ is a map from $\mathcal{M}_{\mathcal{P}}$ to $\mathcal{L}_{\mathcal{H}}$.  Property 1 of gPVMs holds, since
\begin{equation}
\tilde{J}(\mathcal{P}) = J \Big( \bigcup \mathcal{P} \Big) = J(\Omega) = I,
\end{equation}
using that $J$ is a gPVM.  For property 2, consider $\tilde{R},\tilde{S} \subseteq \mathcal{P}$ such that $\tilde{R} \cap \tilde{S} = \emptyset$.  It is easy to see that this implies that $\bigcup \tilde{R} \cap \bigcup \tilde{S} = \emptyset$ (since any two elements of $\mathcal{P}$ are either equal or disjoint as sets), and so we have
\begin{equation}
\tilde{J}(\tilde{R}) = J \Big( \bigcup \tilde{R} \Big) \perp J \Big( \bigcup \tilde{S} \Big) = \tilde{J}( \tilde{S}),
\end{equation}
again using that $J$ is a gPVM.  Finally, we see that equation (\ref{eq:property3}) is satisfies, since for $\tilde{R} \subseteq \tilde{S} \subseteq \mathcal{P}$, we clearly have that $\bigcup \tilde{R} \subseteq \bigcup \tilde{S}$, so that
\begin{equation}
\tilde{J}(\tilde{R}) = J \Big( \bigcup \tilde{R} \Big) \subseteq J\Big( \bigcup \tilde{S} \Big) = \tilde{J}(\tilde{S})
\end{equation} 
since $J$ is a gPVM, and so $\tilde{J}$ is a gPVM by Lemma \ref{lem:char}.  

Next, we assume that there is a countable subset $\mathcal{P}_0 \in \mathcal{M}_{\mathcal{P}}$ such that equation (\ref{eq:sum_to_I}) holds.  Consider some countable set $\{\tilde{Q}_i \}_{i=1}^\infty \subseteq \mathcal{M}_{\mathcal{P}}$.  Defining $Q_i = \bigcup \tilde{Q}_i$  (for all $i \in \mathbb{N}$), note that for each $E \in \mathcal{P}_0$ and each $Q_i$, we clearly have $Q_i \cap E = E$  or $Q_i \cap E = \emptyset$.  Then we see immediately (using Lemma \ref{lem:partition_sets_sums}) that
\begin{equation}
\bigvee_{i=1}^\infty \tilde{J} ( \tilde{Q}_i ) = \bigvee_{i=1}^\infty J(Q_i) = J \Big( \bigcup_{i=1}^\infty Q_i \Big) = \tilde{J} \Big( \bigcup_{i=1}^\infty \tilde{Q}_i \Big),
\end{equation}
which shows that $\tilde{J}$ is a PVM.
\end{proof}

\begin{corol}
If $J$ above is a diagonalizable PVM, then $\tilde{J}$ is a diagonalizable PVM.
\end{corol}
\begin{proof}
Take $\mathcal{P}_0 = \{ E \in \mathcal{P}  \, : \, E \cap \sigma_p(J) \neq \emptyset \}$. 
\end{proof}

\begin{kap}
Let $(\Omega,\mathcal{M})$ be a measurable space, let $\mathcal{P}$ be a partition of $\Omega$ such that $\mathcal{P} \subseteq \mathcal{M}$, and let $\mathcal{M}_{\mathcal{P}}$ be as in Definition \ref{kap:CGmeasurablespace} above.  If 
\begin{equation}
\label{eq:partition_good}
 \{ X \in \mathcal{P} \, : \, E \cap X \neq \emptyset \} \in \mathcal{M}_{\mathcal{P}}
\end{equation}
whenever $E \in \mathcal{M}$, then we call $\mathcal{P}$ an \emph{appropriate partition of $\Omega$}.
\end{kap}

Note that whenever $\mathcal{P}$ is a countable set, it is necessarily an appropriate partition.

\begin{thrm}
\label{thm:coarse}
Let $A,B$ be gPVMs on a measurable space $(\Omega, \mathcal{M})$, and let $P_A$ and $P_B$ be appropriate partitions of $\Omega$.  Further let $P$ be the partition of $\Omega^2$ given by $P:= \{ p \times q \ : \ p \in P_A, \ q \in P_B \}$.  Then
\begin{equation}
\label{eq:cgWB}
J_{\tilde{A} \tilde{B}} = \tilde{J}_{AB} ,
\end{equation}
where $\tilde{A}$, $\tilde{B}$, and $\tilde{J}_{AB}$ are defined (via equation (\ref{eq:coarse})) with respect to the partitions $P_A$, $P_B$, and $P$, respectively.
\end{thrm}
\begin{proof}
For this proof we will need two simple results whose proofs we omit since they use only elementary set theory --- first, for any sets $X$ and $Y$ whose elements are sets, and such that $X \subseteq Y$, we have $\bigcup X \subseteq \bigcup Y$.  Also, if we have two sets $R \subseteq P_A$ and $S \subseteq P_B$, then we have
\begin{equation}
\bigcup ( R \times S) = \Big( \bigcup R \Big) \times \Big( \bigcup S \Big) .
\end{equation}

Consider any $\tilde{Q} \in \mathcal{M}_P$.  We will prove the above result by showing that $\tilde{J}_{AB}(\tilde{Q}) \leq J_{\tilde{A} \tilde{B} } (\tilde{Q} )$ and also  $J_{\tilde{A} \tilde{B} } (\tilde{Q} ) \leq \tilde{J}_{AB}(\tilde{Q})$.  Now, we have both that
\begin{equation}
\label{eq:tildeJ}
\tilde{J}_{AB} (\tilde{Q} ) = J_{AB} \Big( \bigcup \tilde{Q} \Big) = \bigvee_{\substack{R_1 \times R_2 \subseteq \bigcup \tilde{Q} \\ R_1, R_2 \in \mathcal{M} } } A(R_1) \wedge B(R_2),
\end{equation}
as well as
\begin{equation}
\label{eq:tildeAB}
J_{\tilde{A} \tilde{B} } (\tilde{Q} ) = \bigvee_{\substack{\tilde{R}_1 \times \tilde{R}_2 \subseteq \tilde{Q} \\ \tilde{R}_1 \times \tilde{R}_2 \in \mathcal{M}_{P}}} \tilde{A} (\tilde{R}_1 ) \wedge \tilde{B} ( \tilde{R}_2) .
\end{equation}

Considering any element in the join of equation (\ref{eq:tildeAB}), we see that since $\tilde{R}_1 \times \tilde{R}_2 \subseteq \tilde{Q}$, we must have
\begin{equation}
\bigcup \tilde{R}_1 \times \bigcup \tilde{R}_2 = \bigcup (\tilde{R}_1 \times \tilde{R}_2) \subseteq \bigcup \tilde{Q},
\end{equation}
and moreover, by the definition of $\mathcal{M}_P$, we must have $\bigcup \tilde{R}_1 \times \bigcup \tilde{R}_2 \in \mathcal{M}^2$, so that both $\bigcup \tilde{R}_1, \bigcup \tilde{R}_2 \in \mathcal{M}$.  From this we immediately conclude that indeed $J_{\tilde{A} \tilde{B} } (\tilde{Q} ) \leq \tilde{J}_{AB}(\tilde{Q})$, since each element in the join of equation (\ref{eq:tildeAB}) occurs in the join in equation (\ref{eq:tildeJ}) (taking $R_i = \bigcup \tilde{R}_i$ for $i =1,2$).  

For the other inequality, consider some term in the join of equation (\ref{eq:tildeJ}), and then for the $R_1, R_2$ occurring in this term define $\tilde{R}_1$ to be the set of all $X \in P_A$ such that $x \in R_1$ for some $x \in X$, and similarly for $\tilde{R}_2$.  First, we will show that $\tilde{R}_1 \times \tilde{R}_2 \subseteq \tilde{Q}$, so consider some $X \in \tilde{R}_1$ and $Y \in \tilde{R}_2$.  Then we must have $(a,b) \in R_1 \times R_2$ for some $a \in X$ and $b \in Y$.   Since $R_1 \times R_2 \subseteq \bigcup \tilde{Q}$, we must have $(a,b) \in \bigcup \tilde{Q}$, so that there exists some $T \in \tilde{Q}$ with $(a,b) \in T$.  Then, since $\tilde{Q}  \in \mathcal{M}_P$, we must have that $\tilde{Q} \subseteq P$, and hence that $T \in P$.  But also, $X \times Y \in P$, and since $(a,b)$ is a common element of $X \times Y$ and $T$, and $P$ is a partition, we must have $T = X \times Y$, so that $X \times Y \in \tilde{Q}$, and hence that $\tilde{R}_1 \times \tilde{R}_2 \subseteq \tilde{Q}$.

Now, we also have that $ \tilde{R}_1 \in \mathcal{M}_{P_A}$ by equation (\ref{eq:partition_good}), and similarly $\tilde{R}_2 \in \mathcal{M}_{P_B}$, so that $\tilde{R}_1 \times \tilde{R}_2 \in \mathcal{M}_P$.  Then we have that $\tilde{A}( \tilde{R}_1) \wedge \tilde{B}( \tilde{R}_2 ) $ is in the join of equation (\ref{eq:tildeAB}), and also that
\begin{align}
A (R_1) \wedge B( R_2) & \leq A \Big( \bigcup \tilde{R}_1 \Big) \wedge B \Big( \bigcup \tilde{R}_2 \Big) \nonumber \\
& \qquad  = \tilde{A} ( \tilde{R}_1 ) \wedge \tilde{B} ( \tilde{R}_2 ), 
\end{align}
since $R_i \subseteq \bigcup \tilde{R}_i$ for $i =1,2$.  Hence we have established $J_{\tilde{A} \tilde{B} } (\tilde{Q} ) \geq \tilde{J}_{AB}(\tilde{Q})$, since each element in the join of equation (\ref{eq:tildeJ}) is greater than an element in the join in equation (\ref{eq:tildeAB}).
\end{proof}



\subsection{Characterization of $J_{AB}$} 

Before proving our main results, we will need two useful lemmas.  First, recall that for a diagonalizable PVM $A$, the spectrum $\sigma(A)$ is a compact set which is the closure of the set of eigenvalues $\sigma_p(A)$, and moreover, there exists an orthonormal basis for $\mathcal{H}$ consisting of eigenvectors of $A$.  For such an $A$, and for any $R \in \mathcal{B}(\mathbb{R})$, we have that $A(R) = A(R \cap \sigma_p(A))$.  From this fact we immediately deduce the following lemma.

\begin{lemmy}
\label{lem:spectrum}
Let $A,B$ be diagonalizable PVMs.  Then 
\begin{equation}
J_{AB}(Q) = \hspace{-20pt} \bigvee_{\substack{R_1 \times R_2 \subseteq Q \\ R_1,R_2 \in \mathcal{B}(\mathbb{R}) \\ R_1 \subseteq \sigma(A), R_2 \subseteq \sigma(B)}} \hspace{-20pt} A(R_1) \wedge B(R_2) \quad \forall Q \in \mathcal{B}(\mathbb{R}^2),
\end{equation}
and ${  J_{AB} \big( Q \cap \sigma(A) \times \sigma(B) \big) = J_{AB}(Q) \quad \forall Q \in \mathcal{B}(\mathbb{R}^2) }$.
\end{lemmy}

Now, we also have the following nice behavior of our joint observable with respect to common eigenvectors.

\begin{lemmy}
\label{lem:eigen}
Let $A,B$ be diagonalizable PVMs, and let $\ket{\psi} \in \mathcal{H}$ be such that $A \ket{\psi} = a \ket{\psi}$ and $B \ket{\psi} = b \ket{\psi}$.  Then for any $Q \in \mathcal{B}(\mathbb{R}^2)$, we have that
\begin{enumerate}[(1)]
\item $(a,b) \in Q$ implies that $J_{AB}(Q) \ket{\psi} = \ket{\psi}$;
\item $(a,b) \in Q^c$ implies that $J_{AB}(Q) \ket{\psi} = 0$.
\end{enumerate}
\end{lemmy}
\begin{proof}
For the first statement, note that by our assumption we have $A( \{a \} ) \ket{\psi} = \ket{\psi}$ and also $B ( \{ b \} ) \ket{\psi} = \ket{\psi}$, and this means that $A( \{ a \} ) \wedge B( \{ b \} ) \ket{\psi} = \ket{\psi}$.  Then, if $(a,b) \in Q$, since $J_{AB}$ is a gPVM, we have that 
\begin{equation}
J_{AB}(Q) \geq J_{AB}(\{(a,b) \}) = A(\{ a \} ) \wedge B( \{ b \}),
\end{equation}
establishing (1) above.  If, on the other hand, we assume that $(a,b) \in Q^c$, then we have that $J_{AB} (Q^c) \ket{\psi} = \ket{\psi}$ by (1), and $J_{AB}(Q) \perp J_{AB}(Q^c)$ since $J_{AB}$ is a gPVM, which gives that $J_{AB}(Q) \ket{\psi} = 0$.
\end{proof}

We now present the characterization theorem for our generalized joint observables.

\begin{thrm}
\label{thm:characterization}
Let $A,B$ be diagonalizable PVMs, and let $J:\mathcal{B}(\mathbb{R}^2) \to \mathcal{L}_{\mathcal{H}}$ be a set map.  Then $J = J_{AB}$ if and only if $J$ satisfies the following two conditions for all ${ Q \in \mathcal{B}(\mathbb{R}^2) }$ and all $\ket{\psi} \in \mathcal{H}$.
\begin{enumerate}[(1)]
\item If there exist $R_1,R_2 \in \mathcal{B}(\mathbb{R})$ with $R_1 \times R_2 \subseteq Q$ and $A(R_1) \wedge B(R_2) \ket{\psi} = \ket{\psi}$, then $J(Q) \ket{\psi} = \ket{\psi}$.
\item If for every $R_1,R_2 \in \mathcal{B}(\mathbb{R})$ with $R_1 \times R_2 \subseteq Q$ we have $A(R_1) \wedge B(R_2) \ket{\psi} = 0$, then $J(Q) \ket{\psi} = 0$.
\end{enumerate}
\end{thrm}
\begin{proof}
First we will show that condition 1 above implies that $J \geq J_{AB}$, so consider any $Q \in \mathcal{B}(\mathbb{R}^2)$, and any $R_1,R_2 \in \mathcal{B}(\mathbb{R})$ such that $R_1 \times R_2 \subseteq Q$.  Now if ${ A(R_1) \wedge B(R_2) \ket{\psi} = \ket{\psi} }$, by assumption we must have ${ J(Q) \ket{\psi} = \ket{\psi} }$, i.e.~${ J(Q) \geq A(R_1) \wedge B(R_2) }$.  Since this is true for any such ${ R_1 \times R_2 \subseteq Q }$, taking the join gives that ${ J(Q) \geq J_{AB}(Q) }$.  Since this is true for any ${ Q \in \mathcal{B}(\mathbb{R}^2) }$, it follows that ${ J \geq J_{AB} }$.

Next we show that $J \geq J_{AB}$ implies condition 1 above.  Given $Q \in \mathcal{B}(\mathbb{R}^2)$ and $\ket{\psi} \in \mathcal{H}$, assume that there exist ${ R_1,R_2 \in \mathcal{B}(\mathbb{R}) }$ with ${ R_1 \times R_2 \subseteq Q }$ such that ${ A(R_1) \wedge B(R_2) \ket{\psi} = \ket{\psi} }$.  Then we have 
\begin{equation}
J(Q) \geq J_{AB}(Q) \geq A(R_1) \wedge B(R_2),
\end{equation}  
so that $J(Q) \ket{\psi} = \ket{\psi}$.

Finally, we show that condition 2 above is equivalent to $J \leq J_{AB}$.  Note that $J \leq J_{AB}$ if and only if ${ J(Q)^\perp \geq J_{AB}(Q)^\perp }$ for all ${ Q \in \mathcal{B}(\mathbb{R}^2) }$.  Now assume condition 2 and consider ${ \ket{\psi} \in \mathcal{H} }$ such that ${ J_{AB}(Q) \ket{\psi} = 0 }$.  By definition of the join, we must have ${ A(R_1) \wedge B(R_2) \ket{\psi} = 0 }$ for all ${ R_1,R_2 \in \mathcal{B}(\mathbb{R}) }$ such that ${ R_1 \times R_2 \subseteq Q }$.  Then condition 2 gives that ${ J(Q) \ket{\psi} = 0 }$, which shows that ${ J \leq J_{AB} }$.  

On the other hand, given $Q \in \mathcal{B}(\mathbb{R}^2)$, assume that ${ J(Q)^\perp \geq J_{AB}(Q)^\perp }$, as well as that the hypothesis of condition 2 holds for a given ${ \ket{\psi} \in \mathcal{H} }$.  We can easily see that ${ J_{AB}(Q) \ket{\psi} = 0 }$, so that ${J(Q) \ket{\psi} = 0 }$.  As such, condition 2 holds. 
\end{proof}

\subsection{Functional Calculus}

We now prove a generalization of the result that our $f(A,B)$ (previously defined in section \ref{sec:funcs}) is in fact a gPVM.

\begin{thrm}
\label{thm:f_pvm}
Let $(\Omega, \mathcal{M})$ and $(\Omega', \mathcal{N})$ be measurable spaces, let $f:\Omega \to \Omega'$ be a measurable function, and let $J$ be a gPVM on $(\Omega, \mathcal{M})$.  Then $J' := J \circ f^{-1}$ is a gPVM on $(\Omega', \mathcal{N})$.  
\end{thrm}
\begin{proof}
Since $f$ is measurable, we have that $f^{-1} (Q) \in \mathcal{M}$ whenever $Q \in \mathcal{N}$, so that $J': \mathcal{N} \to \mathcal{L}_{\mathcal{H}}$.  Then 
we have
\begin{equation}
J' (\Omega ') = J \circ f^{-1} ( \Omega ') = J (\Omega) = I .
\end{equation}
Next, if $R, S \in \mathcal{N}$ are disjoint, we have that 
\begin{equation}
f^{-1}(R) \cap f^{-1} (S) = f^{-1} (R \cap S) = f^{-1} (\emptyset) = \emptyset 
\end{equation}
and since $J$ is a gPVM, this means that $J' (R) \perp J'(S)$.  Finally, for $R,S \in \mathcal{N}$ with $R \subseteq S$, we have that $f^{-1} (R) \subseteq f^{-1} (S)$, so that $J' (R) \leq J'(S)$, which again, follows from the fact that $J$ is a gPVM.  
\end{proof}

As such, for given diagonalizable PVMs $A$ and $B$, and Borel measurable function $f:\mathbb{R}^2 \to \mathbb{R}$,
\[ f(A,B) := J_{AB} \circ f^{-1} \]
is a gPVM.

In the sequel, we will find the property of $f(A,B)$ in the following lemma useful.

\begin{lemmy}
\label{lem:smsets}
Let $A,B$ be diagonalizable PVMs, and let $f:\mathbb{R}^2 \to \mathbb{R}$ be Borel measurable.  Then, for all $R \in \mathcal{B}(\mathbb{R})$, 
\begin{equation}
f(A,B) \big( R \cap f(\sigma(A), \sigma(B)) \big) = f(A,B)(R). 
\end{equation}
\end{lemmy}
\begin{proof}
To reduce notational clutter, let ${ S := f(\sigma(A), \sigma(B)) }$.  By Lemma \ref{lem:spectrum}, along with the fact that $f^{-1} \circ f(X) \supseteq X$ for any set $X$, we have that
\begin{align}
f(A,B) (E \cap S) & = J_{AB} \big( f^{-1}(E) \cap f^{-1}(S) \big) \nonumber \\
&  \geq  J_{AB} \big( f^{-1}(E) \cap \big( \sigma(A) \times \sigma(B) \big) \big) \nonumber \\
&  \quad = J_{AB} \circ f^{-1} (E) = f(A,B) (E), \label{eq:f_is_big}
\end{align}
and so equality holds (since $f(A,B)$ is a gPVM, and property 2 of Lemma \ref{lem:gpvm_props} gives the other inequality). 
\end{proof}

We now prove that we can construct PVMs out of these gPVMs. We first make the following definition.

\begin{kap}
Let $\mathcal{M}$ be a Boolean $\sigma$-algebra, and let $\mathcal{E} \subseteq \mathcal{M}$ be a subset of $\mathcal{M}$ which is totally ordered under inclusion and which furthermore generates $\mathcal{M}$ as a $\sigma$-algebra.  Then $\mathcal{E}$ will be called a \emph{generating chain} for $\mathcal{M}$.  If $\mathcal{M} = \mathcal{B}(\mathbb{R})$, we will simply refer to $\mathcal{E}$ as a generating chain. 
\end{kap}

\begin{thrm}
\label{thm:chain_works_J}
Let $(\Omega, \mathcal{M})$ be a measurable space, let $\mathcal{E}$ be a generating chain for $\mathcal{M}$, and let $J$ be a gPVM on $(\Omega, \mathcal{M})$ such that there exists a finite set $S \in \mathcal{M}$ which satisfies $J(Q \cap S) = J(Q)$ for all $Q \in \mathcal{M}$.  Then $J|_{\mathcal{E}}$ uniquely extends to a PVM on $(\Omega, \mathcal{M})$\footnote{For any map $f:X \to Y$, and any $Z \subseteq X$, we use the notation $f|_Z$ to mean the map $f$ restricted to the subset~$Z$.}. 
\end{thrm}
\begin{proof}
The result will follow from Sikorski's extension theorem --- first, let $\alpha: \mathcal{E} \to \{-1,1\}$ be any set map, and for any $E \in \mathcal{M}$ define $1 \cdot E := E$ and $-1 \cdot E := E^c$, while for any $P \in \mathcal{L}_{\mathcal{H}}$ define $1 \cdot P := P$ and $-1 \cdot P := P^\perp$.    We will show that for any countable $\mathcal{E}_0 \subseteq \mathcal{E}$, such that 
\begin{equation}
\bigcap_{E \in \mathcal{E}_0} \alpha(E) \cdot E = \emptyset,
\end{equation}
we have that
\begin{equation}
\bigwedge_{E \in \mathcal{E}_0} \alpha(E) \cdot J(E) = 0.
\end{equation}
First, let $\mathcal{E}_+ := \alpha^{-1} (\{ 1\} ) \cap \mathcal{E}_0$ and ${\mathcal{E}_- := \alpha^{-1}(\{ -1 \} ) \cap \mathcal{E}_0}$.  Since $\mathcal{E}$ is a generating chain for $\mathcal{M}$, we have that $\bigcap_{E \in \mathcal{E}_+} E = E_+$ as well as $\bigcup_{E \in \mathcal{E}_-} E = E_-$ for some $E_+, E_- \in \mathcal{M}$.  Then
\begin{align}
\bigcap_{E \in \mathcal{E}_0} \alpha (E) \cdot E & = \bigcap_{E \in \mathcal{E}_+}  E \cap \bigcap_{E \in \mathcal{E}_-} E^c \nonumber \\
& = E_+ \cap E_-^c \nonumber \\
& = \emptyset,
\end{align}
so that $E_+ \subseteq E_-$.  Now, for any $Q \in \mathcal{M}$, we have (by assumption) that $J(Q \cap S) = J(Q)$.  Using this, we have
\begin{align}
\bigwedge_{E \in \mathcal{E}_0} \alpha (E) \cdot J(E) & = \bigwedge_{E \in \mathcal{E}_+}  J(E) \wedge \bigwedge_{E \in \mathcal{E}_-} J( E)^\perp \nonumber \\
& = \bigwedge_{E \in \mathcal{E}_+}  J(E \cap S) \wedge \bigwedge_{E \in \mathcal{E}_-} J( E \cap S )^\perp \nonumber \\
& =  J ( E_+ \cap S ) \wedge J( E_- \cap S)^\perp  \nonumber \\
& = J( E_+ ) \wedge J( E_- )^\perp  , \label{eq:jsig}
\end{align}
where the second to last equality holds because $J$ is monotonic and $S$ is a finite set, so that $ \{ E \cap S \ : \ E \in \mathcal{E}_0 \}$ is a finite set as well.  Now, since $E_+ \subseteq E_-$, we must have that $J(E_+) \leq J( E_-)$, and hence that $J (E_+) \perp J( E_-) ^\perp$, which gives (combining with equation (\ref{eq:jsig}))
\begin{equation}
\bigwedge_{E \in \mathcal{E}_0} \alpha (E) \cdot J(E) = 0.
\end{equation}
The result then follows directly from (one version of) the Sikorski extension theorem (theorem 34.1 in \cite{Sikorski}).  Uniqueness follows trivially from the fact that any two such extensions must agree on a generating chain, and hence must be equal. 
\end{proof}

\begin{corol}
\label{cor:chain_works}
Let $\mathcal{E}$ be a generating chain, let $A,B$ be diagonalizable PVMs with finite spectra, and let $f:\mathbb{R}^2 \to \mathbb{R}$ be Borel measurable.  Then $f(A,B)|_{\mathcal{E}}$ uniquely extends to a PVM. 
\end{corol}
\begin{proof}
This follows directly from the above theorem and Lemma \ref{lem:smsets}, since $f(A,B)$ is a gPVM and we can take $S = f \big( \sigma(A), \sigma(B) \big)$.
\end{proof}

Given two diagonalizable PVMs $A$ and $B$ with finite spectra, a generating chain $\mathcal{E}$, and a Borel measurable function $f:\mathbb{R}^2 \to \mathbb{R}$, we will denote the PVM agreeing with $f(A,B)|_{\mathcal{E}}$ (provided by corollary \ref{cor:chain_works} above) as $f_{\mathcal{E}}(A,B)$.  Note that $f_{\mathcal{E}}(A,B)(Q) = 0$ for all $Q \subseteq f(\sigma(A), \sigma(B))^c$.  We also have the following result which goes beyond the finite spectra case.

\begin{thrm}
Let $A,B$ be diagonalizable PVMs, and let $f:\mathbb{R}^2 \to \mathbb{R}$ be continuous.  Then
\begin{equation}
E_\lambda := J_{AB} \circ f^{-1} \big((- \infty, \lambda] \big)
\end{equation}
is a spectral family of projectors on $\mathcal{H}$. 
\end{thrm}
\begin{proof}
First note that $E_\lambda$ is clearly a projection operator for all $\lambda$, and that $\lambda \mapsto E_\lambda$ is obviously a monotone map (i.e.~$\lambda_1 \leq \lambda_2$ implies $E_{\lambda_1} \leq E_{\lambda_2}$). Next we show that $\lim_{\lambda \to -\infty} E_\lambda = 0$ and $\lim_{\lambda \to \infty} E_\lambda = I$.  Since $f$ is continuous and $A$ and $B$ are diagonalizable, ${ f(\sigma(A), \sigma(B)) }$ is a compact subset of $\mathbb{R}$, so let $m$ denote the minimum and $M$ the maximum of this set.  By Lemma \ref{lem:spectrum} we immediately see that
\begin{align}
E_{m-1} & =  \hspace{-10pt} \bigvee_{\substack{R_1 \times R_2 \subseteq f^{-1}((-\infty, m-1]) \\ R_1,R_2 \in \mathcal{B}(\mathbb{R}) \\ R_1 \subseteq \sigma(A), R_2 \subseteq \sigma(B)}} \hspace{-20pt} A(R_1) \wedge B(R_2) \nonumber \\
 & = \bigvee \emptyset = 0
\end{align} 
and 
\begin{align}
E_{M} & =  \hspace{-10pt} \bigvee_{\substack{R_1 \times R_2 \subseteq f^{-1}((-\infty, M]) \\ R_1,R_2 \in \mathcal{B}(\mathbb{R}) \\ R_1 \subseteq \sigma(A), R_2 \subseteq \sigma(B)}} \hspace{-20pt} A(R_1) \wedge B(R_2) \nonumber \\
 & \geq A(\sigma(A)) \wedge B(\sigma (B)) = I,
\end{align} 
and since $E_\lambda$ is an increasing function of $\lambda$, this demonstrates the desired property.
\end{proof}

The content of the above lemma is that for $A,B$ diagonalizable (but not necessarily with finite spectra), for the particular generating chain 
\begin{equation}
\label{eq:intervals}
\mathcal{E}^{\star} := \{ ( -\infty, \lambda ] \ : \ \lambda \in \mathbb{R} \},
\end{equation}
we can construct a PVM which agrees with $f(A,B)$ on $\mathcal{E}^{\star}$.  In what follows, for $A,B$ diagonalizable PVMs, we will call a generating chain $\mathcal{E}$ \emph{appropriate} if there is a PVM $f_{\mathcal{E}}(A,B)$ which agrees with $f(A,B)$ on $\mathcal{E}$.  It is easy to see that any such PVM must be unique.

We now show that the assumption that the generating set is a chain is necessary in order to construct PVMs.

\begin{lemmy}
Let $\mathcal{H}$ be a two-dimensional Hilbert space, and let $A = \sigma_x$ and $B = \sigma_y$ (Pauli matrices), and assume that $\mathcal{E}$ $\sigma$-generates $\mathcal{B}(\mathbb{R})$, but that $\mathcal{E}$ is not a chain under $\subseteq$.  Then there exists a continuous (and hence Borel measurable) $f:\mathbb{R}^2 \to \mathbb{R}$ such that $J_{AB} \circ f^{-1}|_{\mathcal{E}}$ does not extend to a PVM with its spectra contained in $f(\sigma(A), \sigma(B))$.
\end{lemmy}
\begin{proof}
Since $\mathcal{E}$ is not a chain, there exists some $E_1, E_2 \in \mathcal{E}$ such that there is some $\alpha \in E_1 \cap E_2^c$ and some $\beta \in E_2 \cap E_1^c$.  Note that $\sigma(A) \times \sigma(B) = \{ (1,1), (1,-1), (-1,1), (-1,-1) \}$.  By the Tietze extension theorem, there exists a continuous $f:\mathbb{R}^2 \to \mathbb{R}$ such that $f(1,1) = f(-1,-1) = \alpha$ and $f(-1,1) = f(1,-1) = \beta$.

We now prove the lemma by contradiction, so assume that there is a PVM $F$ extending $J_{AB} \circ f^{-1}$ which is such that $\sigma(F) \subseteq f(\sigma(A), \sigma(B)) = \{ \alpha, \beta \}$.  Then we must have that
\begin{equation}
 F(\{ \alpha, \beta \}) = I . 
\end{equation}
However, since $\beta \subseteq E_2^c$, we must have $E_2 \subseteq \{ \beta \}^c$, and hence $f^{-1}(E_2) \subseteq f^{-1}(\{ \beta \})^c$, and in particular $f^{-1} (E_2) \cap \big( \sigma(A) \times \sigma(B) \big) \subseteq \{ (1,1), (-1,-1) \}$.  From this we immediately see that 
\begin{align}
J_{AB} \circ f^{-1} (E_2) & = J_{AB} \Big( f^{-1} (E_2) \cap \big( \sigma(A) \times \sigma(B) \big) \Big) \nonumber \\
& \leq  J_{AB} \big( \{ (1,1), (-1,-1) \} \big) \nonumber \\
& \quad = 0. 
\end{align} 
A similar argument shows that $J_{AB} \circ f^{-1} (E_1) = 0$.  But, since $F$ is a PVM extending $J_{AB} \circ f^{-1}|_{\mathcal{E}}$, we also have
\begin{align}
 F( \{ \alpha, \beta \} ) & = F(\{ \alpha \}) \vee F( \{ \beta \}) \leq F (E_1) \vee F(E_2) \nonumber \\
& = J_{AB} \circ f^{-1} (E_1) \vee J_{AB} \circ f^{-1} (E_2) \nonumber \\
& = 0, 
\end{align}
which is the desired contradiction.
\end{proof}

We next prove our results about the properties of an observable $f_{\mathcal{E}}(A,B)$ from section \ref{sec:bp-functionalcalculus}.  We begin with a result which is useful in demonstrating these properties. 

\begin{lemmy}
\label{lem:unitary}
Let $A,B$ be diagonalizable PVMs, let $U$ be a unitary operator on $\mathcal{H}$, and let $Q \in \mathcal{B}(\mathbb{R})$.  Then $U J_{AB} (Q) U^\dag =  J_{ UAU^\dag, \, UBU^\dag} (Q)$.
\end{lemmy}
\begin{proof} 
As is well-known (see \cite{kalmbach}), for any $P,Q \in \mathcal{L}_{\mathcal{H}}$, we have that $U (P \wedge Q) U^\dag = (U P U^\dag) \wedge (U Q U^\dag)$, and also for any collection $\{P_j\}_{j \in J} \subseteq \mathcal{L}_{\mathcal{H}}$, where $J$ is any set, we have 
\begin{equation}
\bigvee_{j \in J} (U P_j U^\dag) = U \Big( \bigvee_{j \in J} P_j \Big) U^\dag .
\end{equation}
Using this, we compute.
\begin{align}
U J_{AB}(Q) U^\dag & =  U \Big( \hspace{-10pt} \bigvee_{\substack{R_1 \times R_2 \subseteq Q \\ R_1,R_2 \in \mathcal{M}}} \hspace{-10pt} A(R_1) \wedge B(R_2) \Big) U^\dag \nonumber \\
& = \hspace{-10pt} \bigvee_{\substack{R_1 \times R_2 \subseteq Q \\ R_1,R_2 \in \mathcal{M}}} \hspace{-10pt} U \big( A(R_1) \wedge B(R_2) \big) U^\dag \nonumber \\
& = \hspace{-10pt} \bigvee_{\substack{R_1 \times R_2 \subseteq Q \\ R_1,R_2 \in \mathcal{M}}} \hspace{-10pt}  \big( U A(R_1) U^\dag \wedge U B(R_2) U^\dag \big) \nonumber \\
& = J_{ UAU^\dag, \, UBU^\dag} (Q).
\end{align} 
\end{proof}


\begin{thrm}
\label{thm:f_intervals}
Let $A,B$ be diagonalizable PVMs, and let $\mathcal{E} = \mathcal{E}^{\star}$ (equation (\ref{eq:intervals})).  Furthermore let $f:\mathbb{R}^2 \to \mathbb{R}$ be continuous.  Then $f_{\mathcal{E}}(A,B)$ satisfies
\begin{enumerate}[(1)]
\item $\sigma(f_{\mathcal{E}}(A,B)) \subseteq f( \sigma(A) , \sigma(B) )$;
\item $f_{\mathcal{E}}(UAU^\dag, UBU^\dag) = U f_{\mathcal{E}}(A,B) U^\dag$ for any unitary operator $U$ on $\mathcal{H}$;
\item If $A \ket{\psi} = a \ket{\psi}$ and $B \ket{\psi} = b \ket{\psi}$ for some $\ket{\psi} \in \mathcal{H}$, then $f_{\mathcal{E}}(A,B) \ket{\psi} = f(a,b) \ket{\psi}$.
\end{enumerate}
\end{thrm}
\begin{proof}
We begin with property 1 above.  Assume that $\lambda \in \sigma(f_{\mathcal{E}}(A,B))$, which is true if and only if the corresponding spectral family is not constant for every interval surrounding $\lambda$.  In this case we must have, for every $n \in \mathbb{N}$, that there exists some $R_1^n,R_2^n \in \mathcal{B}(\mathbb{R}) \cap \big( \sigma(A) \times \sigma(B) \big)$ with $R_1^n \times R_2^n \subseteq f^{-1}((-\infty, \lambda + \frac{1}{n}])$ but $R_1^n \times R_2^n \not \subseteq f^{-1}((-\infty, \lambda - \frac{1}{n}])$.  Hence, for each $n \in \mathbb{N}$ there exists some $(a_n, b_n) \in \sigma(A) \times \sigma(B)$ with $f(a_n,b_n) \in (\lambda - \frac{1}{n}, \lambda + \frac{1}{n}]$, so that $\lim_{n \to \infty} f(a_n,b_n) = \lambda$.  Of course, since $\sigma(A) \times \sigma(B)$ is compact, $(a_n,b_n)$ has a convergent subsequence, converging to some $(a,b) \in \sigma(A) \times \sigma(B)$, and since $f$ is continuous, $f(a,b) = \lambda$.

Property 2 above follows directly from Lemma \ref{lem:unitary}, since $f(A,B) = J_{AB} \circ f^{-1}$, and so for any $Q \in \mathcal{E}^{\star}$, we have
\begin{align}
f(U A U^\dag , U B U^\dag ) (Q) & = J_{U A U^\dag, \, U B U^\dag}  \circ f^{-1} (Q) \nonumber \\
& = J_{U A U^\dag, \, U B U^\dag } \big( f^{-1} (Q) \big) \nonumber \\
& = U \big( J_{AB} \big( f^{-1} (Q) \big) \big) U^\dag \nonumber \\
& = U f(A,B) (Q) U^\dag . 
\end{align}
The result then follows from the fact that $f_{\mathcal{E}}(A,B) (Q) = f(A,B) (Q)$ for all $Q \in \mathcal{E}$, that $\mathcal{E}$ generates $\mathcal{B}(\mathbb{R})$, and that conjugation by $U$ induces a $\sigma$-homomorphism on $\mathcal{L}_{\mathcal{H}}$. 

For property 3, given the hypothesis stated above, we must show that $f_{\mathcal{E}}(A,B)( \{f(a,b) \} ) \ket{\psi} = \ket{\psi}$.  Since $f_{\mathcal{E}}(A,B)$ is a PVM, this amounts to showing that $f(A,B)( (-\infty, f(a,b)] ) \ket{\psi} = \ket{\psi}$, while for any $n \in \mathbb{N}$, that $f(A,B)( (-\infty, f(a,b) - \frac{1}{n}] ) \ket{\psi} = 0$.  For the first statement, we have
\begin{align}
f(A,B) ( \{ f(a,b) \} ) & = J_{AB} \circ f^{-1} ( \{ f(a,b) \} ) \nonumber \\
& =  \hspace{-10pt} \bigvee_{\substack{R_1 \times R_2 \subseteq f^{-1} (f (a,b)) \\ R_1,R_2 \in \mathcal{M}}} \hspace{-10pt} A(R_1) \wedge B(R_2)  \nonumber \\
& \quad \geq A( \{ a \} ) \wedge B ( \{ b \} ) , \label{eq:fisbig}
\end{align}
since $\{ a \} , \{ b \} \in \mathcal{B}(\mathbb{R})$.  Now, by assumption we have both $A ( \{ a \} ) \ket{\psi} = \ket{\psi}$ and $B (\{ b \}) \ket{\psi} = \ket{\psi}$, and hence we also have $A (\{ a \} ) \wedge B( \{ b \} ) \ket{\psi} = \ket{\psi}$.  Then, using equation (\ref{eq:fisbig}) and the fact that $f(A,B)$ is a gPVM, we have that
\begin{align}
f(A,B)( (-\infty, f(a,b)] ) & \geq f(A,B) ( \{f(a,b) \} ) \nonumber \\
& \geq A( \{a \} ) \wedge B ( \{ b \} ) .
\end{align}
As such, it follows that $f(A,B) ( ( -\infty, f(a,b) ] ) \ket{\psi} = \ket{\psi}$.

Now, consider any $n \in \mathbb{N}$.  Since we have that ${ ( -\infty, f(a,b) - \frac{1}{n}] \cap \{ f(a,b) \} = \emptyset }$, and $f(A,B)$ is a gPVM, we have that $f(A,B) ( ( -\infty , f(a,b) - \frac{1}{n} ] ) \perp f(A,B) (\{ f(a,b) \} )$, and since $f(A,B) ( \{f(a,b) \} ) \ket{\psi} = \ket{\psi}$ by above argument, we must have $f(A,B) (( -\infty , f(a,b) - \frac{1}{n} ] ) \ket{\psi} = 0$ for any $n \in \mathbb{N}$.
\end{proof}

Analogs of the above results also hold in a slightly different context.

\begin{thrm}
\label{thm:f_nice}
Let $A,B$ be diagonalizable PVMs with finite spectra, and let $\mathcal{E}$ be a generating chain.  Furthermore let $f:\mathbb{R}^2 \to \mathbb{R}$ be Borel measurable.  Then $f_{\mathcal{E}}(A,B)$ satisfies
\begin{enumerate}[(1)]
\item $\sigma(f_{\mathcal{E}}(A,B)) \subseteq f( \sigma(A) , \sigma(B) )$;
\item $f_{\mathcal{E}}(UAU^\dag, UBU^\dag) = U f_{\mathcal{E}}(A,B) U^\dag$ for any unitary operator $U$ on $\mathcal{H}$;
\item If $A \ket{\psi} = a \ket{\psi}$ and $B \ket{\psi} = b \ket{\psi}$ for some $\ket{\psi} \in \mathcal{H}$, then $f_{\mathcal{E}}(A,B) \ket{\psi} = f(a,b) \ket{\psi}$.
\end{enumerate}
\end{thrm}
\begin{proof}
To reduce clutter, we define $S := f(\sigma(A), \sigma(B))$ and $\mathcal{E}' := \{ E \cap S \ : \ E \in \mathcal{E} \}$, and we note that since $A$ and $B$ have finite spectra, both $S$ and $\mathcal{E}'$ are finite sets.

First, we consider property 1, so assume that ${ \lambda \in \sigma \big( f_{\mathcal{E}}(A,B) \big) }$, so that 
\begin{equation}
f_{\mathcal{E}}(A,B) ( \{ \lambda \})  \neq 0.
\end{equation}
Then 
\begin{equation}
f_{\mathcal{E}} (A,B) ( \{ \lambda \} ) \leq f_{\mathcal{E}}(A,B) (E)
\end{equation}
for any $E \in \mathcal{E}$ such that $\lambda \in E$, and also
\begin{equation}
f_{\mathcal{E}} (A,B) ( \{ \lambda \} ) \leq f_{\mathcal{E}}(A,B) (E^c) = f(A,B)(E)^\perp
\end{equation}
for any $E \in \mathcal{E}$ such that $\lambda \in E^c$, since $f_{\mathcal{E}}(A,B)$ is a PVM (which agrees with $f(A,B)$ on elements of $\mathcal{E}$).  This then yields
\begin{equation}
\label{eq:nonzero}
\bigwedge_{ \substack{ E \in \mathcal{E} \\ \lambda \in E }} f(A,B)(E) \wedge \bigwedge_{ \substack{ E \in \mathcal{E} \\ \lambda \notin E }} f(A,B )(E)^\perp \neq 0.
\end{equation}
However, by Lemma \ref{lem:smsets}, we have that $f(A,B) (E \cap S) = f(A,B) (E)$.
This means that
\begin{equation}
\bigwedge_{ \substack{ E \in \mathcal{E}' \\ \lambda \in E }} f(A,B)(E) \wedge \bigwedge_{ \substack{ E \in \mathcal{E}' \\ \lambda \notin E }} f(A,B)(E)^\perp \neq 0.
\end{equation}
Then, since $\mathcal{E}'$ is a finite set, by property 4 in Lemma \ref{lem:gpvm_props}, we have
\begin{equation}
f(A,B) \Big( \bigcap_{ \substack{ E \in \mathcal{E}' \\ \lambda \in E }} E   \Big) \wedge  \Big[  f(A,B) \Big( \bigcup_{ \substack{ E  \in \mathcal{E}' \\ \lambda \notin E }} E  \Big) \Big]^\perp \neq 0,
\end{equation}
which means that 
\begin{equation}
\bigcap_{ \substack{ E \in \mathcal{E}' \\ \lambda \in E }}  E    \neq \bigcup_{ \substack{ E  \in \mathcal{E}' \\ \lambda \notin E }} E,
\end{equation}
or, equivalently
\begin{equation}
\bigcap_{ \substack{ E  \in \mathcal{E} \\ \lambda \in E }} ( E \cap S ) \neq \bigcup_{ \substack{ E  \in \mathcal{E} \\ \lambda \notin E }} (E \cap S).
\end{equation}
Of course, since $\mathcal{E}$ is a chain, if we have $E, \hat{E} \in \mathcal{E}$ with $\lambda \in E$ and $\lambda \notin \hat{E}$, then clearly $E \not \subseteq \hat{E}$, and so therefore $\hat{E} \subseteq E$.  For this reason we have that 
\begin{equation}
\bigcap_{ \substack{ E  \in \mathcal{E} \\ \lambda \in E }} ( E \cap S ) \supseteq \bigcup_{ \substack{ E  \in \mathcal{E} \\ \lambda \notin E }} (E \cap S),
\end{equation}
so that there must be some $\eta$ satisfying 
\begin{equation}
\eta \in \bigcap_{ \substack{ E  \in \mathcal{E} \\ \lambda \in E }} ( E \cap S ) \quad \textrm{but} \quad \eta \notin \bigcup_{ \substack{ E  \in \mathcal{E} \\ \lambda \notin E }} (E \cap S) = \Big(  \bigcup_{ \substack{ E  \in \mathcal{E} \\ \lambda \notin E }} E \Big) \cap S ,
\end{equation}
and hence, since $\eta \in S$, 
\begin{equation}
\eta \notin \bigcup_{ \substack{ E  \in \mathcal{E} \\ \lambda \notin E }} E , \quad \textrm{i.e.} \quad \eta \in \bigcap_{ \substack{ E  \in \mathcal{E} \\ \lambda \in E^c }} E^c.
\end{equation}
If $\lambda \notin S$, then clearly $\eta \neq \lambda$, but this means that, for all $E \in \mathcal{E}$, that $\eta \in E$ if and only if $\lambda \in E$.  However, a simple inductive argument then yields that $\mathcal{E}$ could not then generate $\mathcal{B}(\mathbb{R})$, and hence we must have that $\eta = \lambda$, so that $\lambda \in S$.

The proof of property 2 is similar to the analogous statement in Theorem \ref{thm:f_intervals}, so we omit it.

For property 3, we need to show that $f_{\mathcal{E}}(A,B)( \{f(a,b) \} ) \ket{\psi} = \ket{\psi}$.  Since $\mathcal{E}$ generates $\mathcal{B}(\mathbb{R})$, it is straightforward to show that $\mathcal{E}'$ generates $S$ as an algebra, and so we have that 
\begin{equation}
\{ f(a,b) \} = \Big( \bigcap_{ \substack{ E  \in \mathcal{E}' \\ f(a,b) \in E }} E \Big) \ \cap \ \Big( \bigcap_{ \substack{ E  \in \mathcal{E}' \\ f(a,b) \notin E }} E^c \Big).
\end{equation}
Hence, since $f_{\mathcal{E}}(A,B)$ is a PVM, using considerations similar to those in the proof of property 1, we have that 
\begin{align}
f_{\mathcal{E}} & (A,B)  ( \{f(a,b) \} ) = \nonumber \\
& \bigwedge_{ \substack{ E  \in \mathcal{E}' \\ f(a,b) \in E }} f(A,B)(E) \ \ \wedge \bigwedge_{ \substack{ E  \in \mathcal{E}' \\ f(a,b) \notin E }} f(A,B)(E)^\perp. \label{eq:f_eigen}
\end{align}
Now for any $E \in \mathcal{E}'$, we have $f(A,B) (E) = J_{AB} \circ f^{-1} (E)$ by definition, and also that $f(a,b) \in E$ if and only if $(a,b) \in f^{-1}(E)$.  Hence by Lemma \ref{lem:eigen}, we have that $f(A,B) (E) \ket{\psi} = \ket{\psi}$ for each $E$ with $f(a,b) \in E$, and also that $f(A,B)(E) \ket{\psi} = 0$ for each $E \in \mathcal{E}'$ with $f(a,b) \notin E$, which is to say that $f(A,B)(E)^\perp \ket{\psi} = \ket{\psi}$ for such $E$.  From these considerations, along with equation (\ref{eq:f_eigen}), we then deduce that 
\begin{equation}
f_{\mathcal{E}}  (A,B)  ( \{f(a,b) \} ) \ket{\psi} = \ket{\psi}.
\end{equation}
\end{proof}

We next consider the the properties which are used to establish the equality in expression (\ref{genBCH}) in section \ref{sec:bp-functionalcalculus}.

\begin{lemmy}
\label{lem:composition}
Let $A,B$ be diagonalizable PVMs, let $\mathcal{E}$ be an appropriate generating chain, and let $f:\mathbb{R}^2 \to \mathbb{R}$ and $g:\mathbb{R} \to \mathbb{R}$ be Borel measurable, with $g$ satisfying $g^{-1}(E) \in \mathcal{E}$ for all $E \in \mathcal{E}$.  Then
\begin{equation}
(g \circ f)_{\mathcal{E}}(A,B) = g ( f_{\mathcal{E}} (A,B)) .
\end{equation}
\end{lemmy}
\begin{proof}
We know that for any $Q \in \mathcal{E}$ we have 
\begin{align}
[(g \circ f)_{\mathcal{E}}(A,B)] (Q) & = [J_{AB} \circ (g \circ f)^{-1}](Q) \nonumber\\
& = [J_{AB} \circ f^{-1}](g^{-1}(Q)) \nonumber \\ 
& = f_{\mathcal{E}}(A,B) \circ g^{-1} (Q) \nonumber \\
& = g(f_{\mathcal{E}}(A,B))(Q) .
\end{align}
\end{proof}

\begin{lemmy}
\label{lem:right_comp}
Let $f:\mathbb{R}^2 \to \mathbb{R}$ and $g_1,g_2:\mathbb{R} \to \mathbb{R}$ be Borel measurable functions, let $\mathcal{E}$ be an appropriate generating chain, and let $A,B$ be diagonalizable PVMs.  Then for $h:\mathbb{R}^2 \to \mathbb{R}$ defined by (for all $x,y \in \mathbb{R}$) $h(x,y) := f(g_1(x),g_2(y))$ (i.e.~ $h = f \circ (g_1 \times g_2)$) we have that
\begin{equation}
h_{\mathcal{E}}(A,B) = f_{\mathcal{E}}(g_1(A), g_2(B)).
\end{equation}
\end{lemmy}
\begin{proof}
Since $A,B$ have pure point spectra, we have (for any $Q \in \mathcal{E}$)
\begin{align}
h(A,B) (Q) & = \hspace{1pt} \bigvee_{\substack{R_1 \times R_2 \subseteq h^{-1} (Q) \\ R_1\times R_2 \in \mathcal{B}(\mathbb{R}^2)  }} \hspace{-3pt} A(R_1) \wedge B(R_2) \nonumber \\ 
& = \hspace{-10pt} \bigvee_{\substack{R_1 \times R_2 \subseteq h^{-1} (Q) \\ R_1\times R_2 \subseteq \sigma_p(A) \times \sigma_p(B)  }} \hspace{-13pt} A(R_1) \wedge B(R_2) \nonumber \\ 
& =  \hspace{-12pt} \bigvee_{\substack{g_1(R_1) \times g_2(R_2) \subseteq f^{-1} (Q) \\ R_1 \times R_2 \subseteq \sigma_p(A) \times \sigma_p(B) }} \hspace{-13pt} A(R_1) \wedge B(R_2) .
\end{align}
This follows from 
\begin{equation}
h^{-1}(Q)  = \{ (a,b) \ : \ f(g_1(a),g_2(b)) \in Q \},
\end{equation} so that $(a,b) \in h^{-1}(Q)$ if and only if $(g_1(a),g_2(b)) \in f^{-1}(Q)$, and hence $R_1 \times R_2 \subseteq h^{-1}(Q)$ if and only if $g_1(R_1) \times g_2(R_2) \subseteq f^{-1} (Q) $.

Now, for $R_1 \times R_2 \subseteq \sigma_p(A) \times \sigma_p(B)$ such that ${ g_1(R_1) \times g_2(R_2) \subseteq f^{-1} (Q) }$, we know that  ${ g_1(R_1) \subseteq g_1 (\sigma_p (A) ) }$ and ${ g_2(R_2) \subseteq g_2(\sigma_p(B))}$.  Hence, using that ${ R \subseteq ( g^{-1} \circ g ) (R) }$ we have
\begin{equation}
 h(A,B) (Q) \leq \hspace{-30pt} \bigvee_{\substack{S_1 \times S_2 \subseteq f^{-1} (Q) \\ S_1 \times S_2 \subseteq g_1(\sigma_p(A)) \times g_2(\sigma_p(B)) }} \hspace{-40pt} A(g_1^{-1}(S_1)) \wedge B(g_2^{-1}(S_2)).
\end{equation}
Also, for any $S_1 \times S_2 \subseteq g_1(\sigma_p(A)) \times g_2(\sigma_p(B))$, we clearly have $g_1^{-1}(S_1), g_2^{-1}(S_2) \in \mathcal{B}(\mathbb{R})$.  Additionally, since ${ (g \circ g^{-1}) (S) = S }$, we have
\begin{align}
\bigvee_{\substack{S_1 \times S_2 \subseteq f^{-1} (Q) \\ S_1 \times S_2 \subseteq g_1(\sigma_p(A)) \times g_2(\sigma_p(B)) }}   \hspace{-24pt} & A(g_1^{-1}(S_1)) \wedge B(g_2^{-1}(S_2)) \nonumber \\
& \leq \hspace{-20pt} \bigvee_{\substack{g_1(R_1) \times g_2(R_2) \subseteq f^{-1} (Q) \\ R_1 \times R_2 \in \mathcal{B}(\mathbb{R}^2) }} \hspace{-25pt} A(R_1) \wedge B(R_2) \nonumber \\
& \quad = h(A,B)(Q).
\end{align}
This gives that
\begin{equation}
\hspace{-40 pt} \bigvee_{\substack{S_1 \times S_2 \subseteq f^{-1} (Q) \\  \hspace{40pt} S_1 \times S_2 \subseteq g_1(\sigma_p(A)) \times g_2(\sigma_p(B)) }} \hspace{-55pt} A(g_1^{-1}(S_1)) \wedge B(g_2^{-1}(S_2)) = h(A,B)(Q)
\end{equation}
for all $Q \in \mathcal{E}$.  We also have that
\begin{align}
\bigvee_{\substack{S_1 \times S_2 \subseteq f^{-1} (Q) \\ S_1 \times S_2 \subseteq g_1(\sigma_p(A)) \times g_2(\sigma_p(B)) }} \hspace{-40pt} & A(g_1^{-1}(S_1)) \wedge B(g_2^{-1}(S_2)) \nonumber \\ 
& = \hspace{-30pt} \bigvee_{\substack{S_1 \times S_2 \subseteq f^{-1} (Q) \\ S_1 \times S_2 \subseteq \sigma_p(g_1(A)) \times \sigma_p(g_2(B)) }} \hspace{-35pt} g_1(A)(S_1) \wedge g_2(B)(S_2) \nonumber \\
& = f(g_1(A), g_2(B) ) (Q) ,
\end{align}
and since the above holds for all $Q \in \mathcal{E}$, and $\mathcal{E}$ generates $\mathcal{B}(\mathbb{R})$, we have
\begin{equation}
h_{\mathcal{E}}(A,B) = f_{\mathcal{E}}(g_1(A), g_2(B)).
\end{equation}
\end{proof}

Finally, after a useful lemma, we demonstrate that the spectral order (defined in equation (\ref{spectralorder})) is respected by addition (defined relative to the generating chain $\mathcal{E}^{\star}$).

\begin{lemmy}
\label{catfood}
Let $A,B$ be diagonalizable PVMs, and let $\pdot$ and $\tdot$ be as defined in section \ref{sec:bp-functionalcalculus}, where $\mathcal{E} = \mathcal{E}^{\star}$.  Then, for any $\lambda \in \mathbb{R}$, we have
\begin{equation}
(A \pdot B) \big( (-\infty, \lambda] \big) = \bigvee_{a + b = \lambda} A \big( (-\infty, a ] \big) \wedge B \big( (-\infty, b] \big) . 
\end{equation}
If, furthermore, $A$ and $B$ have no non-negative eigenvalues, we also have
\begin{equation}
(A \tdot B) \big( (-\infty, \lambda] \big) =  \bigvee_{\substack{ ab = \lambda \\ a,b \geq 0 }} A \big( (-\infty, a ] \big) \wedge B \big( (-\infty, b] \big) . 
\end{equation}
\end{lemmy}
\begin{proof}
We first define $Q_\lambda := ( - \infty, \lambda ]$ for all $\lambda \in \mathbb{R}$.  Then we need to prove that 
\begin{equation}
\label{eq:twentyone}
(A \pdot B) ( Q_\lambda ) = \bigvee_{a + b = \lambda} A (Q_a ) \wedge B ( Q_b ) . 
\end{equation}
Recalling that $A \pdot B$ is defined by
\begin{equation}
\label{eq:twenty}
(A \pdot B) ( Q_\lambda ) \hspace{5pt} =  \hspace{-10pt}   \bigvee_{\substack{R_1 \times R_2 \subseteq +^{-1} ( Q_\lambda ) \\ R_1,  R_2 \in \mathcal{B}(\mathbb{R})  }} \hspace{-15pt}  A (R_1) \wedge B(R_2) ,
\end{equation}
where `$+$' is thought of as a map from $\mathbb{R}^2 \to \mathbb{R}$.  We also have
\begin{equation}
+^{-1} ( Q_\lambda ) = \{ (a,b) \in \mathbb{R} \ : \ a+b \leq \lambda \}.
\end{equation}
Since, for any $a+b = \lambda$, we have $Q_a \times Q_b \subseteq +^{-1}(Q_\lambda)$, as well as $Q_a, Q_b \in \mathcal{B}(\mathbb{R})$, we immediately conclude that 
\begin{equation}
(A \pdot B) ( Q_\lambda ) \geq \bigvee_{a + b = \lambda} A (Q_a ) \wedge B ( Q_b ) .
\end{equation}
To establish the opposite inequality, consider any $R_1,R_2 \in \mathcal{B}(\mathbb{R})$ such that $R_1 \times R_2 \subseteq +^{-1}( Q_\lambda )$ for a given $\lambda \in \mathbb{R}$.  Defining $\alpha = \sup R_1$ and $\beta = \lambda - \alpha$, we have that $\alpha + \beta = \lambda$.  Moreover, we clearly have $R_1 \subseteq Q_\lambda$, and also $R_2 \subseteq Q_\beta$, since for any $b \in R_2$, we have $\alpha +b \leq \lambda$.  But from this, we see that every term in the join defining $(A \pdot B) ( Q_\lambda)$ (equation (\ref{eq:twenty})) is less than some element in the join occurring in our desired expression (i.e.~equation (\ref{eq:twentyone})), establishing the other inequality.

A similar argument (with some subtleties concerning negative numbers) yields the result for $\tdot$.
\end{proof}

\begin{lemmy}
Let $A$, $B$, and $C$ be diagonalizable PVMs, with $A \sqsubseteq B$ (where $\sqsubseteq$ denotes the spectral order, as defined as in equation (\ref{spectralorder})).  Also, let $\stackrel{\LARGE \textbf{.}}{+}$ be defined relative to the generating chain $\mathcal{E}^{\star}$.  Then ${ A \stackrel{\LARGE \textbf{.}}{+} C \sqsubseteq B \stackrel{\LARGE \textbf{.}}{+} C }$.   
\end{lemmy}
\begin{proof}
We again let $Q_\lambda := ( - \infty, \lambda ]$ for all $\lambda \in \mathbb{R}$.  Now, by Lemma \ref{catfood}, we have that 
\begin{equation}
(A \stackrel{\LARGE \textbf{.}}{+} C) ( Q_\lambda ) = \bigvee_{a + c = \lambda} A(Q_a) \wedge C(Q_c), 
\end{equation}
as well that 
\begin{equation}
(B \stackrel{\LARGE \textbf{.}}{+} C) (Q_\lambda ) = \bigvee_{b + c = \lambda} B(Q_b) \wedge C(Q_c), 
\end{equation}
and we note that the joins in these expressions run over the same sets since $a+c = \lambda = b+c$ implies $a = b$.  Now, since $A \sqsubseteq B$, we have that $A(Q_\mu) \leq B(Q_\mu)$ for all $\mu \in \mathbb{R}$, from which it follows that 
\begin{equation}
A(Q_\mu) \wedge C(Q_c) \leq B(Q_\mu) \wedge C(Q_c)
\end{equation}
for all $\mu \in \mathbb{R}$.  As such, we have that
\begin{equation}
(A \stackrel{\LARGE \textbf{.}}{+} C) (Q_\lambda ) \leq (B \stackrel{\LARGE \textbf{.}}{+} C) ( Q_\lambda )
\end{equation}
for all $\lambda \in \mathbb{R}$, or equivalently, ${ A \stackrel{\LARGE \textbf{.}}{+} C \sqsubseteq B \stackrel{\LARGE \textbf{.}}{+} C }$.
\end{proof}

\bibliography{bibliography}

\end{document}